\documentclass[a4paper,12pt]{article}
\pdfoutput=1    

\usepackage[a4paper,left=1.2in,right=1.2in,top=1.15in,bottom=0.99in,nohead]{geometry}

\usepackage{setspace}
\usepackage{authblk}
\usepackage{adjustbox}
\usepackage{caption}
\usepackage{subcaption}

\usepackage{rotating} 
\usepackage{csquotes}
\usepackage{amssymb}
\usepackage{lipsum}
\usepackage{amsmath}
\DeclareMathOperator{\sgn}{sgn}
\usepackage{latexsym}
\usepackage{dsfont}
\usepackage{multirow}
\usepackage{floatpag}
\usepackage{eurosym}
\usepackage{graphicx}
\usepackage{url}
\usepackage{natbib}
\usepackage{bbm}
\usepackage{braket}
\usepackage{bm}
\usepackage[section]{placeins}
\usepackage{amsthm}    
\usepackage{enumerate}
\usepackage{pdflscape}     
\usepackage{array}
\usepackage{mathtools}     

\newcolumntype{R}[2]{
    >{\adjustbox{angle=#1,lap=\width-(#2)}\bgroup}
    l
    <{\egroup}
}
\newcommand*\rot{\multicolumn{1}{R{60}{1em}}}

\newcommand{\RN}[1]{
  \textup{\uppercase\expandafter{\romannumeral#1}}
}

\newtheorem{theorem}{Theorem}[section]	
\newtheorem{corollary}[theorem]{Corollary}
\newtheorem{lemma}[theorem]{Lemma}

\theoremstyle{definition}
\newtheorem{example}[theorem]{Example}

\begin{document}

\title{Cross-impact and no-dynamic-arbitrage}

\author[a,b]{\normalsize{Michael Schneider}}  
\author[a,c]{\normalsize{Fabrizio Lillo}}

\affil[a]{Scuola Normale Superiore, Piazza dei Cavalieri 7, 56126 Pisa, Italy}
\affil[b]{Deutsche Bundesbank, Wilhelm-Epstein-Stra{\ss}e 14, 60431 Frankfurt am Main, Germany}
\affil[c]{Department of Mathematics, University of Bologna, Piazza di Porta San Donato 5, 40126 Bologna, Italy}

\maketitle

\begin{abstract}
We extend the ``No-dynamic-arbitrage and market impact''-framework of Jim Gatheral [\textit{Quantitative Finance}, \textbf{10}(7): 749-759 (2010)] to the multi-dimensional case where trading in one asset has a cross-impact on the price of other assets. From the condition of absence of dynamical arbitrage we derive theoretical limits for the size and form of cross-impact that can be directly verified on data. For bounded decay kernels we find that cross-impact must be an odd and linear function of trading intensity and cross-impact from asset $i$ to asset $j$ must be equal to the one from $j$ to $i$.
To test these constraints we  estimate  cross-impact among sovereign bonds traded on the electronic platform MOT. While we find significant violations of the above symmetry condition of cross-impact, we show that these are not arbitrageable with simple strategies because of the presence of the bid-ask spread.
\end{abstract}


{\let\thefootnote\relax\footnotetext{Authors email addresses are michael.schneider@sns.it and fabrizio.lillo@sns.it, respectively. This paper represents the authors' personal opinions and does not necessarily reflect the views of the Deutsche Bundesbank, its staff, or the Eurosystem. 
We thank Borsa Italiana for providing us with access to their datasets and we thank Michael Benzaquen, Jean-Philippe Bouchaud, Katia Colaneri, Thomas Guhr, Florian Kl{\"o}ck, Enrico Melchioni, Loriana Pelizzon, Damian Taranto and participants of Market Microstructure - Confronting Many Viewpoints \#4, Paris, and the $\RN{18}$ Workshop on Quantitative Finance, Milan, for helpful comments. 
We are responsible for all remaining errors.
Michael Schneider acknowledges financial support from the Association of Foundations of Banking Origin.} }

\newpage
\section{Introduction}

Market impact, i.e. the interplay between order flow and price dynamics, has increasingly attracted the attention of researchers and of the industry in the last years (\cite{bouchaud2008markets}). Despite its importance, both from a fundamental point of view (due to its relation with supply-demand) and from an applied point of view (due to its relation with transaction cost analysis and optimal execution), market impact is not yet fully understood and different models and approaches have been proposed and empirically tested. 

It is important to note that market impact refers to different aspects of this interplay and that they should be carefully distinguished (see \cite{bouchaud2008markets} for a discussion). First, there is the impact of an individual trade or of the aggregated signed\footnote{Conventionally buyer (seller) initiated trades have positive (negative) volume and order sign.}  order flow in a fixed time period. Second, especially for transaction cost analysis and optimal execution, it is more interesting to consider the impact of a large trade (sometimes termed as  \textit{meta-order}) executed incrementally by the same investor with many transactions and orders over a given interval of time. Both these definitions of market impact are typically investigated by considering one asset at a time, i.e. without considering the effect of a trade (or of an order) in one asset on the price dynamics of another asset.  

This is the third type of impact, that we study in this paper, and that is termed {\it cross-impact}. Understanding and modeling cross-impact is important for many reasons, since it enters naturally in problems like optimal execution of portfolios, statistical arbitrage of a set of assets, and to study the relation between correlation in prices and correlation in order flows. Conceptually, while self-impact, the impact of a trade on the price of the same asset, can qualitatively be understood as the result of a mechanical component (e.g. a market order with volume larger than the volume at the opposite best) and an induced component (resilience of the order book due to liquidity replenishment), the source of cross-impact is less clear. On one side if a trader is liquidating simultaneously two assets one can obviously expect a non-vanishing cross-impact. Since impact measures are typically averages across many measurements, this mechanism produces cross-impact if simultaneous trades and positively correlated order flow are frequently observed. On the other side, liquidity providers and arbitrageurs detect local mispricing between correlated assets and bet on a reversion to normality by placing orders. In other words this induced cross-impact relates to the possibility of identifying price changes due to local imbalances of supply-demand in one asset (rather than to fundamental information) and of exploiting the possibly short-lived mispricing between correlated assets. 

Even though cross-impact has already been discussed e.g. in \cite{almgren2001optimal} as an extension of their optimal execution model and in \cite{hasbrouck2001common} in a principal component approach, it has only recently been the subject of extensive empirical studies. 
\cite{pasquariello2013strategic} empirically show that order imbalance has a significant impact on returns across stocks and sectors at the daily scale. \cite{wang2016average,wang2016cross} present evidence for a structured price cross-response and correlated order flow at the intraday time-scale across stock pairs. \cite{benzaquen2016dissecting} link cross-response and order flow in a multivariate extension of the Transient Impact Model (TIM) of \cite{bouchaud2004fluctuations} and show that their model can reproduce a significant part of the well-known correlation structure of asset returns. 
\cite{mastromatteo2017trading} exploits this link between correlation and cross-impact, showing that cross-impact is crucial for a correct estimation of liquidity when trading portfolios.
\cite{wang2016microscopic} perform a scenario analysis 
in a model similar to \cite{benzaquen2016dissecting}, finding that cross-response is related both to cross-impact and correlated order flow across assets. 
 
It is clear that the cross-impact problem talks naturally to dynamic arbitrage and to the possibility of price manipulation, as already discussed in \cite{huberman2004price}. 
It is therefore natural to ask which constraints the no-price-manipulation assumption imposes on market impact models. There is a large literature on this problem, often focused on the single asset case (\cite{huberman2004price}, \cite{gatheral2010no}, \cite{alfonsi2012order}, \cite{GatheralHandbook}, \cite{curato2016discrete, curato2016optimal}). 

In the multi-asset case many articles are concerned with strategies for optimal portfolio liquidation in the presence of volatility risk by expanding the model of \cite{almgren2003optimal}. \cite{schied2010optimal} show that optimal execution strategies for investors with constant absolute risk aversion are deterministic and for a more general absolute risk aversion setting \cite{schoneborn2016adaptive} finds that the optimal strategies for investors with different risk preferences vary only in the speed of their execution. The case of cross-impact in a lit market when there is also a dark pool is discussed in \cite{kratz2015portfolio}. 
\cite{tsoukalas2016dynamic} instead develop a limit order book model with cross-impact and find that it can be optimal to temporarily take up positions contrary to the direction of one's trading intent. 
The paper most related to ours from a theoretical point of view is \cite{alfonsi2016multivariate}. They model multi-asset price impact by considering a linear version of the model of \cite{gatheral2010no}, extending thus the model already considered in \cite{alfonsi2012order}. They show that the absence of no-dynamic-arbitrage on a discrete-time grid corresponds to the decay kernel being described by a positive definite matrix-valued function. Furthermore they formulate further conditions to ensure that resulting optimal strategies are well-behaved, both in discrete and continuous time, and show how such kernels can be constructed. However it is not generally straightforward to establish positive definiteness when a decay kernel is obtained coordinate-wise from estimations and therefore necessary conditions for the absence of dynamic arbitrage that can be verified on estimated decay kernels prove useful.

In this paper, focusing on the TIM framework in continuous time, we establish some easily verifiable necessary conditions that must be satisfied by self- and cross-impact, in order to avoid the presence of price manipulation. We do this in the same spirit of \cite{gatheral2010no} by explicitly constructing trading strategies that lead to price manipulation and negative expected cost. Some of these relations are simple generalizations to the multi-asset case of the corresponding relations for the single asset case derived in \cite{gatheral2010no}. Other relations that we derive here are instead genuinely relative to the multi-asset case. In particular we formalize in Lemma \ref{lemma:linbounded_symm} that cross-impact must be symmetric, i.e. the return induced in asset $i$ by a trade of volume $v$ in asset $j$  must be equal to the impact of a trade of the same volume $v$ in asset $i$ on the price of asset $j$. 

It is natural to ask whether this symmetry condition is empirically verified. In this paper we study a market whose microstructure, to the best of our knowledge, has not been explored so far. This is the MOT market for sovereign bonds\footnote{\url{http://www.lseg.com/areas-expertise/our-markets/borsa-italiana/fixed-income-markets/mot}}, a fully electronic limit order book market for fixed income assets. One of the reasons for our choice is that, due to the nature of the traded assets, we expect cross-impact, especially due to quote revisions, to be very high. In fact, two Italian fixed-rate BTPs differ mostly through the coupon rate and the time-to-maturity - factors which are accounted for in the price, which  moves in a very synchronised way since for most purposes both titles are perfectly interchangeable. 

Calibrating a multivariate TIM in trade time we find that there exist pairings of bonds where the symmetry condition of cross-impact is violated in a statistically significant way. By comparing the potential profit from a simple arbitrage strategy to transaction costs such as the bid-ask spread, which are neglected in the model, we conclude that arbitraging is not profitable. It is also crucial to point out that the empirical part of the paper is important because it is the first application of a TIM model to fixed income markets and to the best of our knowledge it is the first work to consider cross-impact of single market orders and not the order sign imbalance aggregated over fixed time intervals (as done in \cite{busseti2011}, \cite{benzaquen2016dissecting} and \cite{wang2016average}).

The rest of the paper is structured as follows. Section \ref{sec:model} introduces our model and the links to the no-dynamic-arbitrage principle. Section \ref{sec:generalconstraints} discusses some general constraints on cross-impact that arise in our framework for bounded decay kernels and the corresponding proofs are given in Appendix \ref{app:proofs}. In Section \ref{sec:empirical} we study cross-impact empirically and compare to the theoretical results in Section \ref{sec:generalconstraints}. Finally Section \ref{sec:conclusion} concludes.

%
\section{Model Setup}
\label{sec:model}

The presence of dynamic arbitrage depends on the market impact model. In this paper we consider the Transient Impact Model (TIM) introduced in \cite{bouchaud2004fluctuations} (see \cite{bouchaud2008markets} for a discussion). The model has been originally formulated in discrete time, and its continuous time version, that we present in the next section, has been proposed in \cite{gatheral2010no}.

\subsection{Price Process and Cost of Trading}

\cite{gatheral2010no} assumes that the asset price $S_t$ at time $t$ follows a random walk with a drift determined by the cumulative effect of previous trades \begin{equation}
S_t = S_0 + \int_0^t{f(\dot{x}_s) G(t-s) \mathrm{d}s} + \int_0^t{\sigma \mathrm{d}Z_s}
\label{eq:price_1D}
\end{equation}
where $f(\dot{x}_s)$ represents the (instantaneous) impact of trading at a rate $\dot{x}_s$ at time $s<t$ weighted by a decay kernel $G(\tau)$ with $\tau = t-s$. $Z_s$ is a noise process, for example a Wiener process, and $\sigma$ is the volatility. For consistency with equation (\ref{eq:cost_ND}) below, the trading rate $\dot{x}$ is given in units of number of shares per unit of time. 
In our multivariate extension we consider the prices of a set of assets where the drift in asset $i$ not only depends on the trading history of asset $i$ but also on past trades in assets $j\neq i$. Thus the price process of asset $i$ is given by
\begin{equation}
S^i_t = S^i_0 + \sum_j{ \int_0^t{f^{ij}(\dot{x}^j_s) G^{ij}(t-s) \mathrm{d}s} } + \int_0^t{\sigma^i \mathrm{d}Z^i_s}
\label{eq:price_ND}
\end{equation}
with a correlated noise process $\bm{Z}_s$ (e.g. a multivariate Wiener process) and where in addition to the \textit{self-impact} terms $f^{ii}$ and $G^{ii}$ we have introduced additive \textit{cross-impact} terms $f^{ij}$ and $G^{ij}$, $i\neq j$, that represent the impact of trading in asset $j$ on the price of asset $i$. 

For a trading strategy $\Pi = \left\{\bm{x}_t \right\}, \: t \in [0,T]$ where $\bm{x}_t$ is the vector of asset positions $x^i_t$ in asset $i$ at time $t$, the expected cost (or implementation shortfall) is
\begin{align}
C(\Pi) &= \mathbb{E}\left[ \sum_i{ \int_0^T{\dot{x}^i_t (S^i_t - S^i_0) \mathrm{d}t}  } \right] \nonumber \\
       &= \sum_{i,j}{ \int_0^T{ \dot{x}^i_t \mathrm{d}t \int_0^t{f^{ij}(\dot{x}^j_s) G^{ij}(t-s) \mathrm{d}s   }  } }
\label{eq:cost_ND}
\end{align}
where as in \cite{gatheral2010no} we  consider only costs due to price impact, i.e. the price shift induced by our own trading, and neglect \emph{slippage costs}, i.e.  costs due to all other market frictions such as bid-ask spreads and trading fees. Also \cite{alfonsi2016multivariate} ignore these costs by arguing that a sophisticated trading strategy consists not only of market orders but also of limit orders, thus on average both paying (earning) a half spread from market (limit) orders. In practice the constraints we derive may well be weakened when slippage costs are unavoidable, e.g. when immediacy requires to execute a trading strategy with market orders.

\subsection{Principle of no-dynamic-arbitrage}

\cite{huberman2004price} define a \textit{round-trip trade} as a sequence of trades whose sum is zero, i.e. a trading strategy $\Pi = \left\{\bm{x}_t \right\}$ with \begin{equation}
\int_0^T{\dot{\bm{x}}_t \mathrm{d}t} = \bm{0}  \: .
\label{eq:roundtripstrategy}
\end{equation}
This implies a round-trip in all assets traded in the strategy, i.e. $\int_0^T{\dot{x}^i_t \mathrm{d}t} = 0 \: \forall \: i$. A \textit{price manipulation} is a round-trip trade $\Pi$ whose expected cost $C(\Pi)$ is negative and the \textit{principle of no-dynamic-arbitrage} states that such a price manipulation is impossible. Formally, the principle requires that for any round trip trade $\Pi$ it is
\begin{equation}
C(\Pi) \geq 0 \: .
\label{eq:nodynamicarb}
\end{equation}
In the one-dimensional case this translates to 
\begin{equation}
C(\Pi) = \int_0^T{ \dot{x}_t \mathrm{d}t \int_0^t{f(\dot{x}_s) G(t-s) \mathrm{d}s   }  } \geq 0
\label{eq:nodynamicarb_1D}
\end{equation}
and imposes a relationship on the market impact function $f(\cdot)$ and the decay kernel $G(\cdot)$. The functions $f(\cdot)$ and $G(\cdot)$ are said to be \textit{consistent} if they exclude the possibility of price manipulation. Several papers have studied the consistency of this kind of market impact models (for a review see \cite{GatheralHandbook}). \cite{gatheral2011exponential} show that any decay kernel that is non-singular at time zero is inconsistent with non-linear $f(\cdot)$. Moreover \cite{gatheral2010no} sets some necessary constraints for no arbitrage for power law dependence of $f$ and $G$ and   \cite{curato2016optimal} show that inconsistencies can also arise for power-law $f(\cdot)$ and $G(\cdot)$ even when necessary conditions derived in \cite{gatheral2010no} are not violated.

In the multidimensional case the cost requirement is
\begin{equation}
C(\Pi) = \sum_{i,j}{ \int_0^T{ \dot{x}^i_t \mathrm{d}t \int_0^t{f^{ij}(\dot{x}^j_s) G^{ij}(t-s) \mathrm{d}s   }  } } \geq 0
\label{eq:nodynamicarb_ND}
\end{equation}
and we are similarly looking at what forms of $\bm{f}(\cdot)$ and $\bm{G}(\cdot)$ are consistent when there is cross-impact.\footnote{\cite{alfonsi2016multivariate} consider a slightly different case where instead of a round-trip strategy, they consider the liquidation of an existing portfolio. This is equivalent in the limit of building up the portfolio infinitely slowly and when impact is purely transient. } 
Specifically we are asking what limits there are to cross-impact, i.e. the form of $f^{ij}(\cdot), i\neq j$, and whether the presence of cross-impact leads to possible arbitrages in pairs of $\bm{f}(\cdot)$ and $\bm{G}(\cdot)$ that are consistent in the one-dimensional case.


%
\section{General constraints on cross-impact for bounded decay kernels}
\label{sec:generalconstraints}

Let us for the remainder of this paper assume, without loss of generality, that $\bm{G}(\tau)$ is a dimensionless quantity, i.e. all dimensionality of cross-impact, including the sign, 
is captured by the instantaneous market impact function $\bm{f}$. 

In this section we also assume that the decay kernel $\bm G(\tau)$ is non-increasing, right-continuous at $\tau = 0$ and bounded in all components, i.e. that there exists an upper bound $U > 0$ so that $|G^{ij}(\tau)| < U$ for all $\tau \in [0,\infty)$ and all $i,j$. 
While we do not consider unbounded kernels in this section, we discuss in Appendix \ref{app:uniquedelta} some constraints that arise for the popular class of pure power law kernels for cross-impact.
The non-increasing assumption rules out non-zero decay kernels with $G^{ij}(0)=0$ so that we are able to take $\bm{G}$ as normalized to $1$ for its smallest lag $\tau = t-s$, i.e. $G^{ij}(0) = 1$ for all pairs $ij$.\footnote{The non-increasing assumption is necessary already in the single-asset case to avoid arbitrage opportunities from simple buy-hold-sell strategies. In the multi-asset case we require it e.g. for our symmetry result in Lemma \ref{lemma:linbounded_symm}. }
A special case of such a kernel is exponential decay $G^{ij}(t-s) = e^{- \rho^{ij} \left( t-s \right)}$ as in \cite{obizhaeva2013optimal}.

In the following we consider the two-dimensional case, $i \in \left\{a,b\right\}$, i.e. the number of assets $N=2$. Note that all results from the one-dimensional case still hold since we are free to choose a trading strategy that is active only in one asset, e.g. $\Pi = \left(\Pi^a, \Pi^b \right)^{\intercal} = \left(\Pi^a, 0 \right)^{\intercal} \:, \: \forall \:t \in [0,T]$ where $\Pi^a$ is a round-trip trading strategy in asset $a$.

The proofs of the results in this section are given in Appendix \ref{app:proofs} and are obtained following the approach of \cite{gatheral2011exponential}. All results can also be obtained following \cite{gatheral2010no} under the slightly more restrictive assumptions of the decay kernel $\bm G$ being representable as a suitable series expansion and considering only the first non-zero orders in the limit of $\tau \rightarrow 0^+$.


\subsection{A simple strategy with two assets}
\label{sec:strategy_easy}

In the following we will often make use of a simple strategy in two assets which is split into two phases of trading at constant rates. 

\begin{example}{A simple in-out strategy.}
\label{eg:strategy_easy}
\\
At first we build up a position at a constant trading rate from time $0$ until time $\Theta$, with $0 < \Theta < T$, and then liquidate the position in a second phase from $\Theta$ until $T$. 

\begin{align}
\Pi = \left\{ {\bm{x}}_t \right\} \quad, \quad \dot{\bm{x}}_t = \begin{cases} 
\left(v_{a,\RN{1}}, v_{b,\RN{1}} \right)^{\intercal} \quad &\text{for} \quad 0 \leq t \leq \Theta \\
\left(v_{a,\RN{2}}, v_{b,\RN{2}} \right)^{\intercal} \quad &\text{for} \quad \Theta < t \leq T 
\end{cases}.
\label{eq:strategy_specific}
\end{align}
The velocities $v_{i,\RN{1}}$, $v_{i,\RN{2}}$ are constrained by our choice of the strategy. Since $\Pi$ is a round-trip strategy, the trading rates $v_{i,\RN{1}}$ and $v_{i,\RN{2}}$ have opposite signs, i.e. $\kappa = v_{i,\RN{1}} / v_{i,\RN{2}} < 0$, and the time $\Theta$ when the trading direction changes is given as $\Theta =  \frac{-v_{i,\RN{2}}}{v_{i,\RN{1}} - v_{i,\RN{2}} }T =  \frac{1}{1 - \kappa}T $. 
Let us further fix notation with $\lambda = v_{a, \RN{1}} / v_{b, \RN{1}} = v_{a, \RN{2}} / v_{b, \RN{2}}$. 
Figure \ref{fig:strategies3_1} illustrates a possible realization of this strategy with $\lambda < 0$. 

The cost of this strategy can be decomposed as $C(\Pi) = \sum_{i,j = a,b}{ C_A^{ij} + C_B^{ij} + C_C^{ij}}$ where 
\begin{align}
C_A^{ij} 	&=  v_{i,\RN{1}} f^{ij}(v_{j,\RN{1}}) \int_0^{\Theta}{\mathrm{d}t \int_0^t{G^{ij}(t-s) \mathrm{d}s }    }    \nonumber  \\
C_B^{ij} 	&=  v_{i,\RN{2}} f^{ij}(v_{j,\RN{1}}) \int_{\Theta}^{T}{\mathrm{d}t \int_0^{\Theta}{G^{ij}(t-s) \mathrm{d}s }    } \label{eq:c21} \\ 
C_C^{ij} 	&=  v_{i,\RN{2}} f^{ij}(v_{j,\RN{2}}) \int_{\Theta}^{T}{\mathrm{d}t \int_{\Theta}^t{G^{ij}(t-s) \mathrm{d}s }    }  \nonumber
\end{align}

In the one-dimensional case the principle of no-dynamic-arbitrage imposes a constraint on the term $C_B^{ii}$ 
and from equation (\ref{eq:nodynamicarb_ND}) it follows that $- C_B^{ii} \leq C_A^{ii} + C_C^{ii}$ as in \cite{gatheral2010no}. 
For the multi-dimensional case $C(\Pi) \geq 0$ further implies a relationship between the strength of cross-impact and self-impact. 
\end{example}

In the following we will try to exploit cross-impact in order to push down the cost of strategies. Supposing that cross-impact is positive for positive trading rates, i.e. $f^{ij}(v) > 0$ for $v > 0$, we can choose $\lambda < 0$, e.g. trading into asset $a$ while contemporaneously trading out of asset $b$, in order to get a negative contribution from cross-impact.

\begin{figure}[!t]
    \centering
        \begin{subfigure}[b]{0.5\textwidth}
                \includegraphics[width=\textwidth]{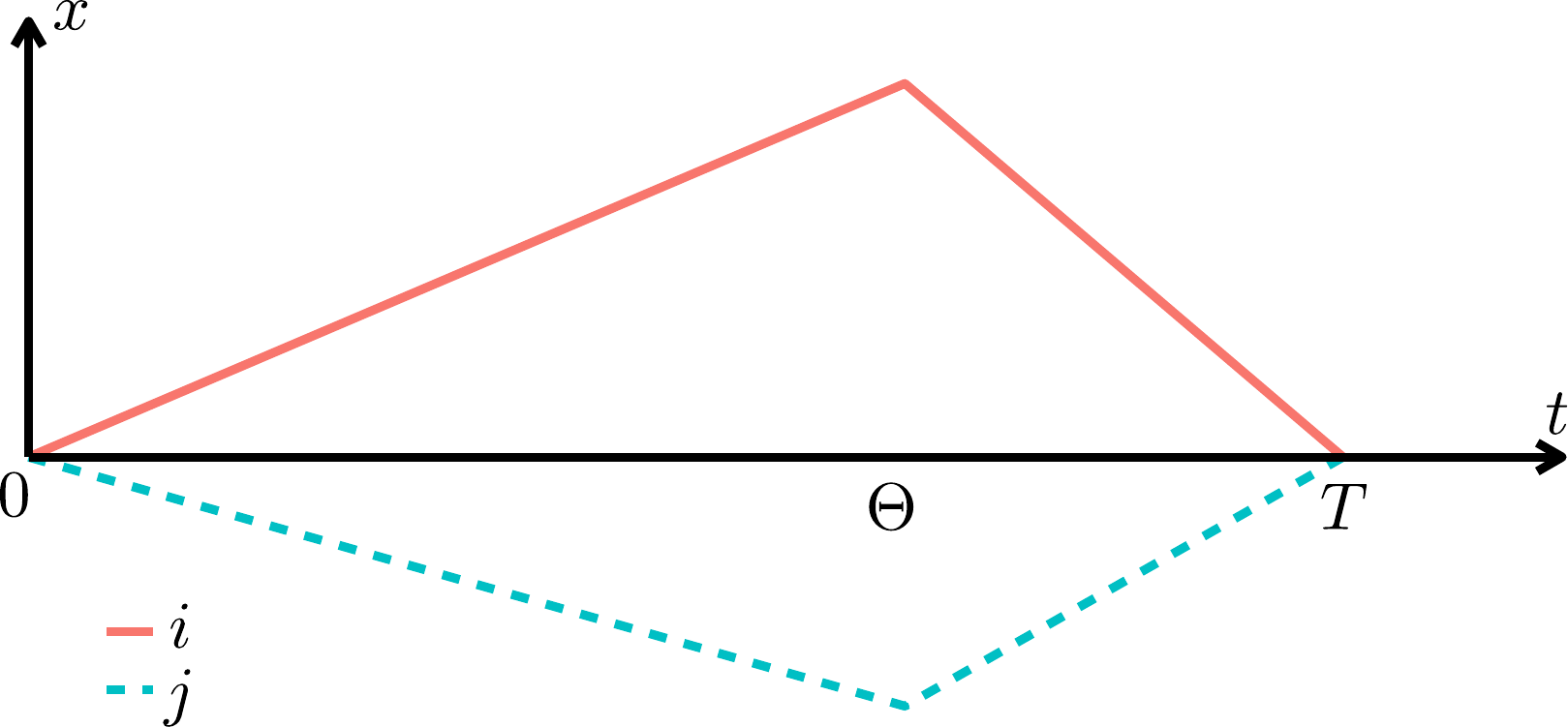}
                \caption{Strategy as in Example \ref{eg:strategy_easy}}
                \label{fig:strategies3_1}
    \end{subfigure}
    
        \begin{subfigure}[b]{0.5\textwidth}
                \includegraphics[width=\textwidth]{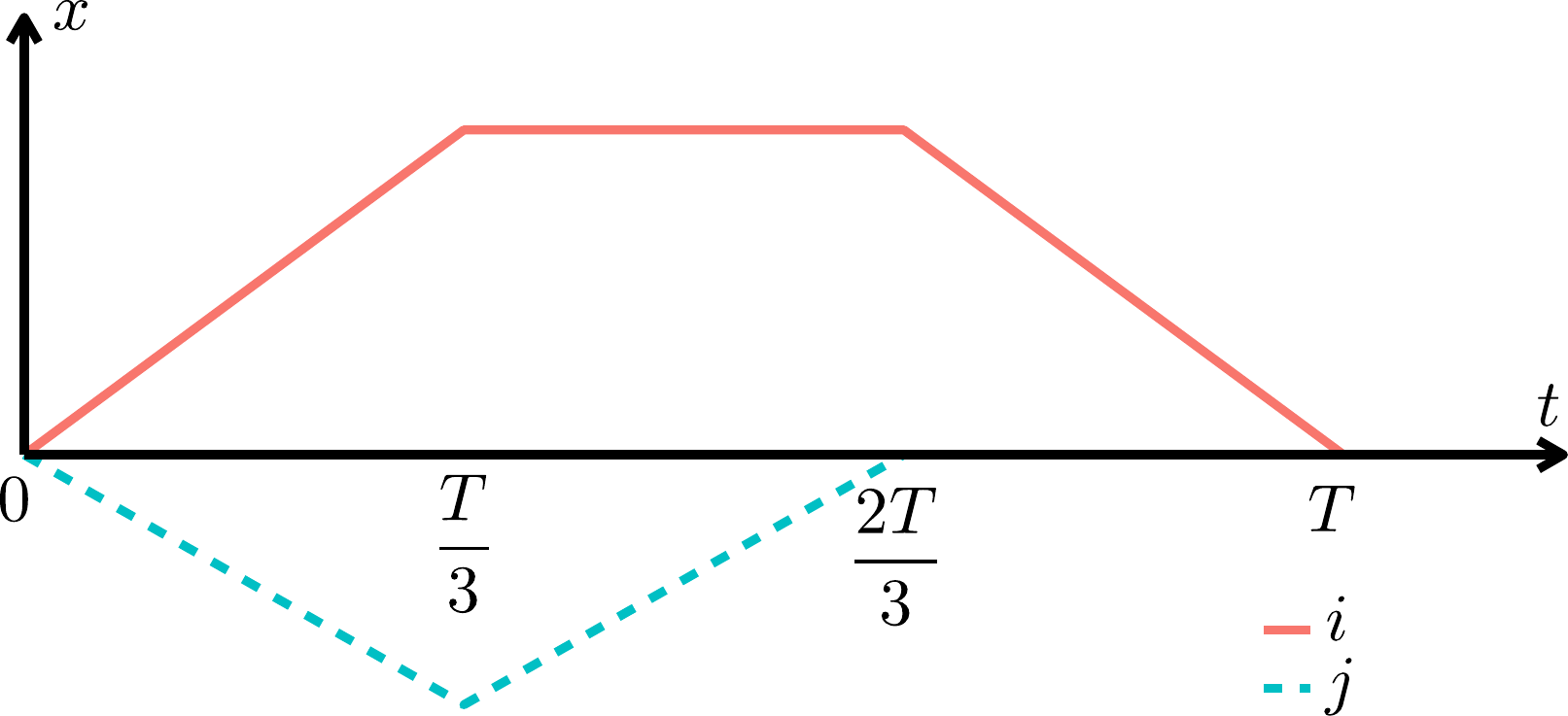}
                \caption{Strategy as in Example \ref{eg:linperm_asymm}}
                \label{fig:strategies3_8}
    \end{subfigure}
    \caption{Schematic of the trading strategies in Example \ref{eg:strategy_easy} (upper panel) and Example \ref{eg:linperm_asymm} (lower panel). }
		\label{fig:strategies} 
\end{figure} 

\subsection{Cross-impact as odd function of the trading rate}
\label{sec:odd}

In the one-dimensional case \cite{gatheral2010no} shows that permanent market impact needs to be an odd function in the rate of trading $v$, i.e. $f(v) = -f(-v)$.
We show here that the same holds for cross-impact for decay kernels that are non-singular around $\tau \rightarrow 0^+$.

\begin{lemma}
\label{lemma:antisym_ND}
Assume a price process as in (\ref{eq:price_ND}) with a bounded, non-increasing decay kernel ${\bm G}$ that is continuous around $\tau = 0$. Then such a model admits price manipulation if ${\bm f}$ is not an odd function of the trading rate, i.e. unless 
\begin{equation}
f^{ij}(v) = -f^{ij}(-v) \: ~~~~~~\forall \: i,j \: .
\label{eq:antisym_ND}
\end{equation}

\end{lemma}

Therefore we will assume for the remainder of this paper that (\ref{eq:antisym_ND}) holds. As a corollary it follows that 

\begin{corollary}
\label{cor:absence}
Absence of dynamic-arbitrage for a price process as in (\ref{eq:price_ND}) with a decay kernel that is bounded, non-increasing and continuous around $\tau = 0$ requires that
\begin{equation}
f^{ij}(v=0) = 0 ~~~~~~ \forall \: i,j \: .
\label{eq:absence}
\end{equation}
\end{corollary}

\subsection{Constraints on the strength of cross-impact}
\label{sec:size}

The cost constraint in equation (\ref{eq:nodynamicarb_ND}) also imposes a constraint on the relative strength of $f^{ij}$. Let us consider a simple example at first. 

\begin{example}{Trading in and out at the same rate}
\label{eg:size_easy}

We consider a strategy as above in Example (\ref{eg:strategy_easy}) where we are trading in and out of positions at the same rate, i.e. $v_{i,\RN{1}} = -v_{i,\RN{2}}$ and therefore $\Theta = T/2$, but in different directions in the two assets, choosing e.g. $v_{a, \RN{1}} = v_a > 0, v_{b, \RN{1}} = -v_b < 0$ and thus $\lambda < 0$. For simplicity let us assume a uniform decay of market impact, i.e. $G^{ij}(t) = G(t)$ for all pairings $ij$. The cost is then
\begin{align}
\label{eq:cost_inoutsame}
C(\Pi) = &\left[ v_a f^{aa}(v_a) + v_b f^{bb}(v_b) - v_a f^{ab}(v_b) - v_b f^{ba}(v_a) \right] \\ \nonumber 
							&\left\{ \int_0^{T/2}{\mathrm{d}t \int_0^t{\left[ G(t-s) - G(t+T/2 - s) \right] \mathrm{d}s } } 
							+ \int_{T/2}^{T}{ \mathrm{d}t \int_{T/2}^{t}{ \left[ G(t-s) - G(T-s) \right] \mathrm{d}s } } \right\}
\end{align}
and \cite{gatheral2010no} shows that the term in curly brackets in equation (\ref{eq:cost_inoutsame}) is greater than zero when further requiring that $G(\cdot)$ is strictly decreasing. Thus the no-dynamic-arbitrage constraint (\ref{eq:nodynamicarb}) requires that
\begin{equation}
v_a f^{aa}(v_a) + v_b f^{bb}(v_b) - v_a f^{ab}(v_b) - v_b f^{ba}(v_a) \geq 0
\label{eq:nodynamicarb_inoutsame}
\end{equation}
for any $v_a, v_b \geq 0$, thus constraining the relative size of the cross-impact terms $f^{ab}$ and $f^{ba}$ with respect to self-impact. 
Note that by setting $v_b = 0$ we recover the one-dimensional case and it follows that $v_a f^{aa}(v_a) \geq 0$.  
\end{example}

In the general case, the decay $G^{ij}(\tau)$ is not uniform and we can not factor out the term in curly brackets in equation (\ref{eq:cost_inoutsame}), instead we have to weight each of the terms of equation (\ref{eq:nodynamicarb_inoutsame}) with a factor that depends on the decay $G^{ij}(\tau)$. Furthermore we are free to choose a strategy with different trading rates as in Example \ref{eg:strategy_easy} or a more sophisticated strategy. \cite{alfonsi2016multivariate} consider this problem in discrete time with linear instantaneous price impact. Their Proposition 2.6 states that absence of arbitrage in the sense of equation (\ref{eq:nodynamicarb_ND}) is equivalent to the condition that the elementwise product of strength of impact and the decay kernel corresponds to a positive definite matrix-valued function.

\subsection{Linearity of market impact}
\label{sec:bounded_linearity}

\cite{gatheral2011exponential} finds that any single-asset market impact model as in equation (\ref{eq:price_1D}) is inconsistent when $f$ is non-linear and $G$ is bounded and non-increasing. We expand this proposition to the multi-asset case with cross-impact, i.e. 
\begin{lemma}
\label{lemma:finite_linear}
Assuming a price process as in (\ref{eq:price_ND}) with a bounded, non-increasing decay kernel ${\bm G}$ that is continuous around $\tau = 0$ and a non-linear market impact function $\bm{f}$. Then such a model admits price manipulation. 
\end{lemma}

As a corollary of Lemma \ref{lemma:finite_linear} we can extend Lemma 4.1 in \cite{gatheral2010no} for self-impact to the case with cross-impact.
\begin{corollary}
\label{cor:exp_linear}
A price process as in (\ref{eq:price_ND}) with self- and cross-impact that decays exponentially at different rates and instantaneous price impact that is non-linear, admits price manipulations.
\end{corollary}
We obtain the same corollary for purely permanent impact by taking the limit $\rho^{ij} \rightarrow 0^+$ in the case of exponential decay $G^{ij}(t-s) = e^{- \rho^{ij} \left( t-s \right)}$, as already observed in \cite{huberman2004price}:
\begin{corollary}
Nonlinear permanent self- and cross-asset market impact is inconsistent with the principle of no-dynamic-arbitrage.  
\end{corollary}

\subsection{Symmetry of cross-impact}
\label{sec:finite_symmetry}

Again we assume that $\bm{G}(\tau)$ is bounded, non-increasing and continuous around $\tau = 0$ and therefore that impact is linear, i.e. $f^{ij}(v) = \eta^{ij} v$, for otherwise our model is inconsistent, as shown in the previous section. We will further show that in this case impact needs to be symmetric, i.e. $\eta^{ij} = \eta^{ji}$, in order to avoid price manipulations. 
 
\begin{example}{An asymmetric strategy with purely permanent and linear impact.}
\label{eg:linperm_asymm}
\\
Suppose market impact is linear and permanent, i.e. $f^{ij}(v) = \eta^{ij} v$ and $G^{ij}(\tau) = 1 \: \forall \: i,j$. Then the cost of trading in the single-asset case only depends on the initial and final positions ${x}_0$ and ${x}_T$. If there is cross-impact between two or more assets, there is also an interaction term between the trading rates in different assets, i.e.
\begin{align}
\nonumber
\label{eq:cost_lin_perm1}
C(\Pi) 	&=	 \sum_{i,j}{ \int_0^T{v^i_t \mathrm{d}t  \int_0^t{ \eta^{ij} v^j_s \mathrm{d}s } } } \\
				&=	 \sum_i{ \frac{\eta^{ii}}{2} (x^i_T - x^i_0) ^2 + \sum_{j \neq i}{ \eta^{ij} \int_0^T{v^i_t \mathrm{d}t  \int_0^t{  v^j_s \mathrm{d}s } }    }  } 
\end{align}
and while the first sum with the self-impact terms disappears for a round-trip strategy since $\bm{x}_0 = \bm{x}_T$, this is not generally the case for the second sum with the terms due to cross-impact. To see this, let us consider a different round-trip strategy $\Pi$ in two assets, which is now asymmetric and lasts over three phases: 
\begin{align}
v^a_t = \begin{cases} 
v^a \quad &\text{for} \quad 0 \leq t \leq T/3 \\
0 \quad &\text{for} \quad T/3 < t \leq 2T/3 \\
-v^a \quad &\text{for} \quad 2T/3 < t \leq T \\
\end{cases} \quad , \quad 
v^b_t = \begin{cases} 
-v^b \quad &\text{for} \quad 0 \leq t \leq T/3 \\
v^b \quad &\text{for} \quad T/3 < t \leq 2T/3 \\
0 \quad &\text{for} \quad 2T/3 < t \leq T \\
\end{cases} \: ,
\label{eq:strategy_asymm}
\end{align}
as illustrated in Figure \ref{fig:strategies3_8}. 
While self-impact cancels out when we calculate $C(\Pi)$, the asymmetry in our strategy makes for a non-trivial total cost that stems from cross-impact:
\begin{equation}
C(\Pi) = v^a v^b \frac{T^2}{18} (\eta^{ba} - \eta^{ab}) .
\label{eq:cost_asymm_perm}
\end{equation}
If $\eta^{ba} < \eta^{ab}$ this gives a negative cost and likewise when $\eta^{ba} > \eta^{ab}$ by interchanging assets $a \leftrightarrow b$ in the strategy (\ref{eq:strategy_asymm}). Therefore it follows that cross-impact needs to be symmetric with respect to asset pairs in order to exclude arbitrage opportunities, as observed in \cite{huberman2004price}. 

\end{example}

In fact we can expand this result to the transient impact case:

\begin{lemma}
\label{lemma:linbounded_symm}
If decay of market impact $\bm{G}(\tau)$ is bounded, non-increasing and continuous around $\tau = 0$ and $f^{ij}(v) = \eta^{ij} v$ is linear, absence of dynamic arbitrage requires that \begin{equation}
\eta^{ij} = \eta^{ji} \quad \forall \: i,j \: .
\label{eq:linbounded_symm}
\end{equation} 
\end{lemma}

Let us reconsider Example \ref{eg:size_easy} taking into account linearity and symmetry of cross-impact as shown above. In this case Equation (\ref{eq:nodynamicarb_inoutsame}) simplifies to  
\begin{equation}
v^2_a \eta^{aa} + v^2_b \eta^{bb} - 2 v_a v_b \eta^{\text{cross}} \geq 0
\label{eq:nodynamicarb_inoutsame_symmetric}
\end{equation}
and minimizing the cost constrains the strength of cross-impact $\eta^{\text{cross}} = \eta^{ab} = \eta^{ba}$ as 
\begin{equation}
\eta^{\text{cross}} \leq \sqrt{\eta^{aa}\eta^{bb}} \: , 
\label{eq:nodynamicarb_size}
\end{equation}
in agreement with Proposition 3.7.(b) of \cite{alfonsi2016multivariate} and equivalent to the condition for a symmetric $2 \times 2$ matrix to be positive-semidefinite.  

\subsection{Exponential decay}
\label{sec:exponential}

The conditions of linearity and symmetry of cross-impact are necessary for absence of arbitrage, but are they also sufficient? 

\begin{example}{An asymmetric strategy with symmetric, exponentially decaying linear impact.}
\label{eg:linexp_asymm}

Let us re-consider the strategy in Eq. (\ref{eq:strategy_asymm}) with exponentially decaying impact $G^{ij}(t-s) = e^{- \rho^{ij} \left( t-s \right)}$ and a linear instantaneous impact function $f^{ij}(v) = \eta^{ij} v$ that is now symmetric with $\eta^{ab} = \eta^{ba} = \eta^{\text{cross}}$. The cost terms for self-impact are now
\begin{align}
C^{aa} &= \frac{\eta^{aa} v_a^2 }{(\rho^{aa})^2} \left[ - e^{-\rho^{aa} T } +2 e^{-2 \rho^{aa} T / 3} + e^{- \rho^{aa} T / 3} - 2 + \frac{2 \rho^{aa} T }{3} \right]  \nonumber  \\ 
C^{bb} &= \frac{\eta^{bb} v_b^2 }{(\rho^{bb})^2} \left[ - e^{-2 \rho^{bb} T / 3} + 4 e^{- \rho^{bb} T / 3} - 3 + \frac{2 \rho^{bb} T }{3} \right]  \: . 
\label{eq:cost_asymm_self_exp}
\end{align}
and likewise for cross-impact
\begin{align}
C^{ab} &= \frac{\eta^{\text{cross}} v_a v_b }{(\rho^{ab})^2} \left[ -\frac{ \rho^{ab} T}{3} + 2 e^{- \rho^{ab} T / 3} - 3 e^{- 2\rho^{ab} T/3} + e^{- \rho^{ab} T}  \right]          \nonumber \\
C^{ba} &=  \frac{\eta^{\text{cross}} v_a v_b }{(\rho^{ba})^2} \left[ 2 -\frac{ \rho^{ba} T}{3} -3 e^{- \rho^{ba} T / 3} + e^{- 2 \rho^{ba} T/3}  \right]         \: . 
\label{eq:cost_asymm_cross_exp}
\end{align}
When we develop the terms in squared brackets in (\ref{eq:cost_asymm_self_exp}) and (\ref{eq:cost_asymm_cross_exp}) in powers of $\rho^{ij}T$ all terms of order $(\rho^{ij}T)^0$ and $(\rho^{ij}T)^1$ cancel out, while terms proportional to $(\rho^{ij}T)^2$ sum to $0$ thanks to the symmetry of instantaneous cross-impact. The cost to the first non-zero order of $\rho^{ij}T$ is then
\begin{equation}
C(\Pi) = \frac{T^3}{6} \left[ \frac{2}{3} \eta^{aa} v_a^2 \rho^{aa} + \frac{4}{27} \eta^{bb} v_b^2 \rho^{bb} - \frac{5}{27} \eta^{\text{cross}} v_a v_b \left( \rho^{ab} + \rho^{ba} \right) \right]
									+ \sum_{i,j} \mathcal{O}\left( (\rho^{ij})^2 T^4 \right) 
\label{eq:cost_asymm_exp}
\end{equation}
and when $\rho^{ab}$ or $\rho^{ba}$ is large enough\footnote{While keeping the product of $\rho^{ij}T$ small for all $ij$.} compared to the other terms the cost can still be negative. 
For absence of price manipulations we therefore also require further constraints on the speed of decay described by $\bm{G}$. 
\end{example}

This is in agreement with the results of \cite{alfonsi2016multivariate} in discrete time. Their Proposition 3.7 proves that the conditions of symmetry $\eta^{ab} = \eta^{ba}$, and a non-increasing decay kernel, i.e. $\min(\rho^{ab}, \rho^{ba}) \geq \frac{1}{2}(\rho^{aa} + \rho^{bb})$ and $\frac{1}{4}(\eta^{ab}\rho^{ab} + \eta^{ba} \rho^{ba})^2 \leq \eta^{aa}\rho^{aa} \eta^{bb} \rho^{bb}$, are sufficient for the absence of arbitrage. We complement this result in Lemma \ref{lemma:linbounded_symm} by showing that symmetry $\eta^{ij} = \eta^{ji}$ is indeed necessary for any decay kernel $\bm{G}(\tau)$ that fulfills our conditions of being bounded, non-increasing and continuous around $\tau = 0$. Note that this excludes kernels where for some $ij$ $G^{ij}(0)=0$ but $G^{ij}(\tau)>0$ for some $\tau > 0$. Indeed Example 3 in \cite{alfonsi2016multivariate} has a kernel that is asymmetric for $\tau > 0$ and which does not allow price manipulation.

\FloatBarrier
\section{Empirical evidence of cross-impact}
\label{sec:empirical}

\FloatBarrier
\subsection{Market Structure of MOT}
\label{sec:mot}

For the empirical analysis we consider Italian sovereign bonds traded on the retail platform ``Mercato telematico delle obbligazioni e dei titoli di Stato'' (MOT). 
We choose to estimate cross-impact between bonds instead of equities since we expect the strength of cross-impact among sovereign bonds of the same issuing country, especially of similar maturity, to be bigger than the one between e.g. stocks or indices. Sovereign bonds of one country typically have a very similar underlying risk and their prices are implicitly connected via the yield curve, a link that we deem stronger than e.g. a common factor between stocks of the same sector. 

The secondary market for European sovereign bonds is divided into an opaque over-the-counter market (OTC) and an observable exchange-traded market. 
The Italian securities and exchange Commission CONSOB publishes a bi-annual report listing the share in trading of Italian government bonds separated per trading venue\footnote{CONSOB, Bollettino Statistico Nr. 8, March 2016, available at \url{http://www.consob.it/web/area-pubblica/bollettino-statistico}}. For the year 2014 (2015) the share of OTC trading has been $58.8\%$ ($59.1\%$), while $45.6\%$ ($44.8\%$) of trading on platforms took place on the inter-dealer platform MTS. MOT is the third-largest platform by traded value with $8.7\%$ ($8.8\%$) of traded value excluding the OTC market in 2014 (2015). 
Most of the literature for the Italian and European government bonds market focuses on MTS, with the exception of \cite{linciano2014liquidity} who compare the liquidity of dual-listed corporate bonds across MOT and the EuroTLX platform. 
\cite{darbha2013microstructure} review the market microstructure of MTS in the context of the market for European sovereign bonds and discuss several liquidity measures based on the limit order book, trades or bond characteristics. 
They note that MTS ``normally has a few trades per bond per day, even for the most liquid government bonds''. Indeed, due to large minimum sizes, for most titles there is on average less than one transaction per day on MTS, making studies of market impact difficult.
\cite{dufour2012permanent} overcome this issue by building impulse response functions from regressions of returns on order flow at 10 second intervals to study permanent market impact. In a different approach \cite{schneider2016has} use a measure of (virtual) mechanical price impact along with other liquidity measures calculated from the limit order book to detect illiquidity shocks that can be modeled as a self- and cross-exciting Hawkes process in and across Italian sovereign bonds.

Instead in this paper we focus on MOT where we observe a sufficient number of (smaller) trades as well as an active limit order book. 
Italian government bonds are traded on the DomesticMOT segment of MOT where the trading day is divided into an opening auction from 8:00 to 9:00 followed\footnote{The conclusion of contracts from the opening auction happens at a random time between 09:00:00-09:00:59.} by a phase of continuous trading until 17:30. If certain price limits are violated during the continuous trading, a volatility auction phase is initiated for a duration of 10-11 minutes. 
MOT is organized as a continuous double auction where besides market and limit orders also partially hidden ``iceberg orders'', ``committed cross'' orders and ``block trade facilities'' are allowed. While the presence of a specialist or a bid specialist is possible, in practice this is only the case for a subset of financial sector corporate bonds not in our sample. The tick size depends on the residual lifetime and is 1 basis point of nominal size or 0.1 basis points if the residual lifetime is less or equal than two years, corresponding to $1$ or $0.1$ euro cents respectively. 

Our dataset contains all trades and limit order book (LOB) snapshots\footnote{In phases of heavy trading multiple updates of the LOB may be recorded as one update in our data. However there is at least one update per second whenever there are changes to the LOB and in the vast majority of our sample updates are more frequent.} for a selection of 60 ISINs from December 1, 2014 to February 27, 2015 and April 13, 2015 to October 16, 2015 for a total of 194 trading days. For the remainder of this paper we will focus on a set of $N=33$ fixed rate or zero-coupon Italian sovereign bonds listed in Appendix \ref{app:isins} with at least 5,000 trades throughout our sample to ensure sufficient liquidity and statistical significance of our results. To avoid intraday seasonalities we further restrict our data to 10:00 - 17:00 and discard observations when we detect a volatility auction. The average spread is smaller than 10 ticks for most of the bonds with the exception of some very long-term bonds and bonds where the tick size is 0.1 basis points.
More than $92\%$ of the orders in our sample are executed at the corresponding best bid or ask quote\footnote{The remaining $\sim 8\%$ can either be due to orders that were executed across more than one millisecond (so that they are recorded as two or more orders), missed LOB updates or exotic order types.} and thus identified as sell or buy orders respectively, while all other orders are classified according to the algorithm of \cite{lee1991inferring}. 

Let us fix notation for the estimations in the following sections. We consider the log-price $X^i_t = \log(S^i_t)$ of the mid-price of the best bid and ask quote for asset $i$ at time $t$ and calculate the return $r^i_{t,t+\Delta t}$ from time $t$ to time $t + \Delta t$ as $r^i_{t,t+\Delta t} = X^i_{t+\Delta t - \varepsilon} - X^i_{t - \varepsilon}$ for $\varepsilon \rightarrow 0^+$. $\epsilon^i_t$ is the sign of a trade (market order) and $+1$ for a buyer-initiated transaction, $-1$ for a sell, and undefined when there is no trade in the asset $i$ at time $t$. $I^i_t$ is an indicator function that is $+1$ when there is a trade in asset $i$ at time $t$ and 0 otherwise and we consider the product of an undefined trade sign with a $0$ indicator function to be 0 such that the product $\epsilon^i_t I^i_t$ is always defined and one of $\{-1, 0, +1\}$. The size of a trade $V^i_t$ is given as its nominal value in EUR and the price is reported per one asset (or contract) with a face value of 100 EUR. 
Unlike e.g. \cite{benzaquen2016dissecting} we do not de-mean the order sign in order to avoid attributing a price impact to the absence of transactions in a bond in the sense of Corollary \ref{cor:absence}. However we have verified that our results are qualitatively similar when considering de-meaned order signs $\epsilon$ or $\epsilon I$ and de-meaned returns.

\FloatBarrier
\subsection{Response function}

We define the self- and cross-response function $R^{ij}_{\Delta t}$ as the unconditional $\Delta t$-ahead return in asset $i$ controlled for the order sign of asset $j$, i.e.
\begin{equation}
R^{ij}_{\Delta t} = \mathbb{E} \left[ \left( X^i_{t+\Delta t-\varepsilon} - X^i_{t-\varepsilon} \right) \epsilon^j_t I^j_{t} \right] .
\label{eq:response_function_R}
\end{equation}
For $i = j$ we will speak of \textit{self-response} and of \textit{cross-response} for $i \neq j$. Figure \ref{fig:response_function} shows the average self- and cross-response function for all bonds in our sample and their pairings respectively.  
\begin{figure}[!ht]
    \centering
        \includegraphics[width=0.7 \textwidth]{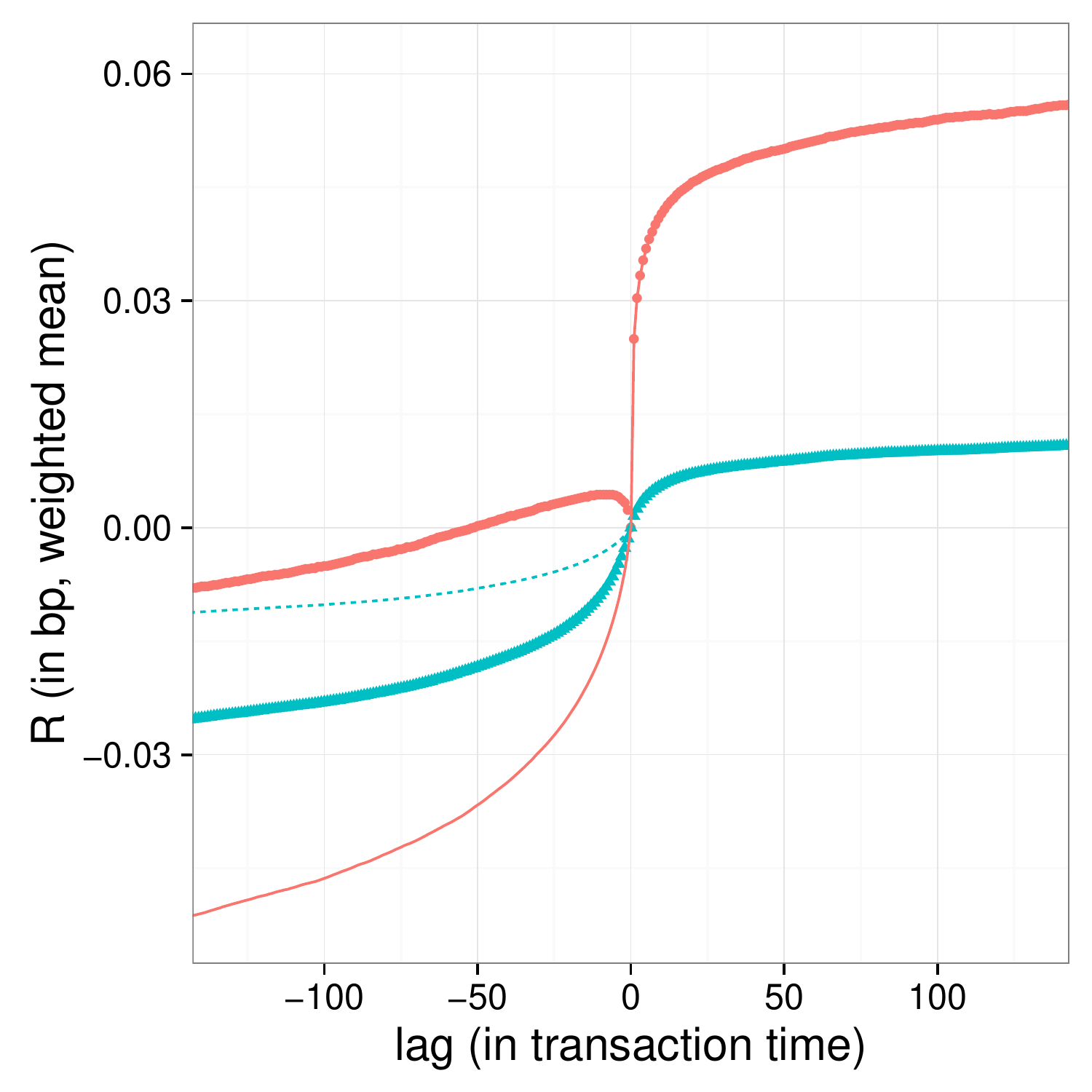}
    \caption{Plot of average self- and cross- response function $R^{ij}_{\Delta t}$ in transaction time as defined in Section \ref{sec:emp_decay}. Mean over all bonds and pairings in our sample and weighted by the number of trades in the triggering bond $j$. Self-response is shown as red dots connected by solid lines, cross-response as blue triangles connected by dashed lines. The lines correspond to the prediction from the model in Section \ref{sec:emp_decay}. }
		\label{fig:response_function} 
\end{figure} 
For positive lags $\Delta t$ we find that self-response is on average larger than cross-response by a factor of $\sim 5$, consistent with observations of \cite{benzaquen2016dissecting, wang2016average, wang2016cross}. $R^{ij}_{\Delta t = 0}$ is zero by definition, whereas for small negative $\Delta t$ we find that $R^{ij}$ is on average positive, producing a cusp at $\Delta t = 0$. We conjecture that such behavior is not observed in \cite{benzaquen2016dissecting} because of the rather large time lag of 5 minutes, corresponding to $\sim 80$ units of transaction time in Figure \ref{fig:response_function}. In the single asset case this feature is clearly present for the large-tick stock Microsoft in Figure 1 of \cite{taranto2016linear}. As shown there, the kink could be related to correlations of market order flow with past returns and indicates a forecasting power of current returns on the future order sign imbalance. Interestingly we find that the cross-response measured at negative lags is smaller (i.e. larger in absolute value) than self-response, contrary to the observations in \cite{benzaquen2016dissecting}.\footnote{We suppose that this is related to the fact that many of the bonds considered here are easily substitutable for one another.} The figure also shows the prediction from the model of the negative lag impact (see \cite{taranto2016linear} for details). We observe a clear difference with the empirical data suggesting also for cross-impact a reaction of order flow to past price dynamics of other bonds.

\FloatBarrier
\subsection{Instantaneous market impact}
\label{sec:emp_f}

We measure the instantaneous market impact function ${\bm f}(\cdot)$ as   \begin{equation}
f^{ij}(V) = \mathbb{E}\left[ r^i_{t-\varepsilon, t+2s} \epsilon^j_{t} | I^j_{t,V} = 1 \right]
\label{eq:f_measurement}
\end{equation}
which is the expected return in asset $i$ from just before a trade at time $t$ until 2 seconds after $t$, multiplied by the trade sign in asset $j$ at time $t$ and conditional on a trade in asset $j$ at time $t$ of size $V$. We have chosen the two second interval as twice the maximum time between two updates of the limit order book, i.e. we can rule out that changes in the book were not reported in our data.
For measurement purposes we bin similar trade sizes together, with the bin size chosen as a function of the number of trades in the triggering bond $j$. 
\begin{figure}[!ht]
    \centering
        \includegraphics[width=\textwidth]{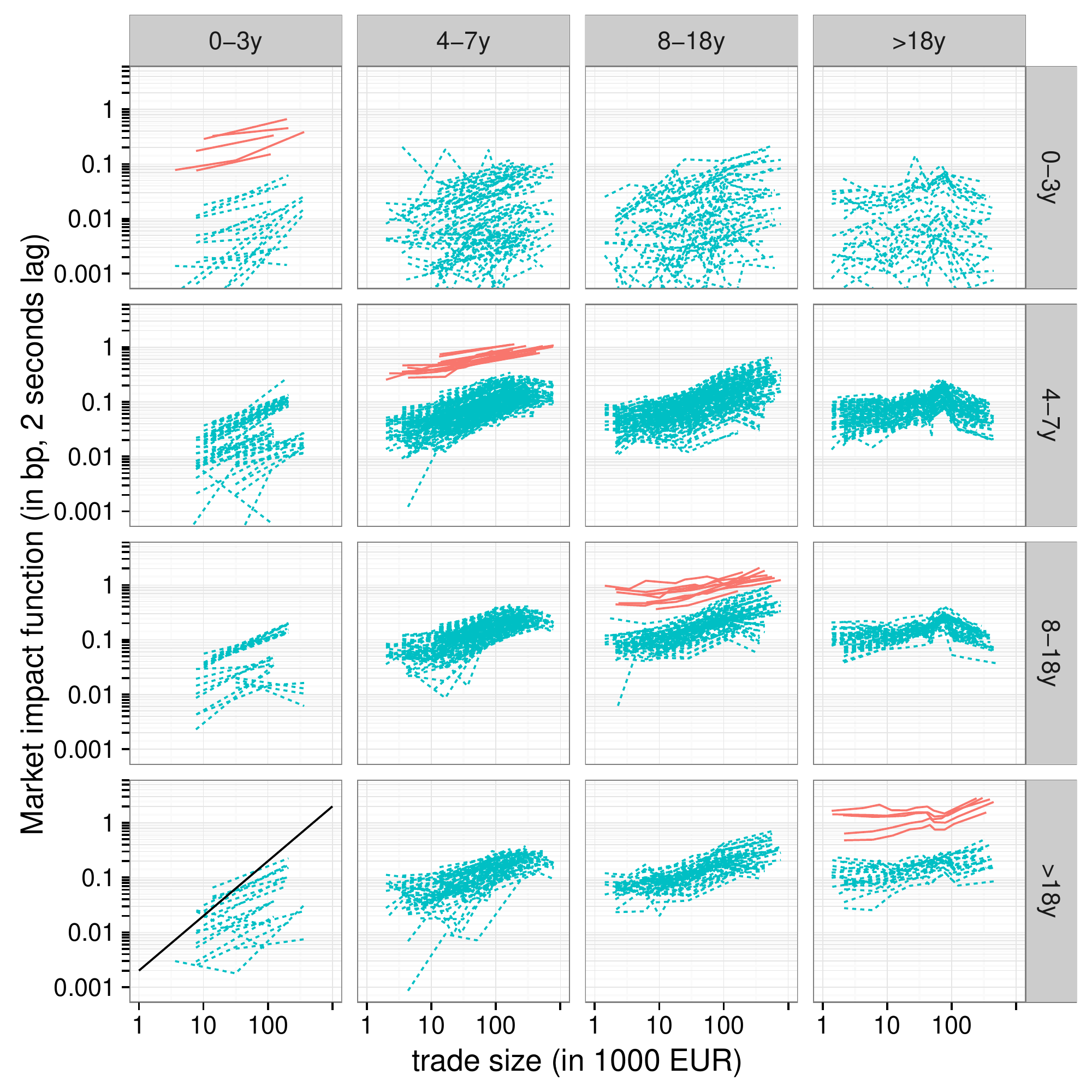}
		\caption{Plot of the average self- and cross-impact function among all pairs of bonds in our sample as a function of trade size $V$ measured in units of face value. Each line corresponds to one pairing $ij$, grouped by time-to-maturity into four categories, where impact is from the column on the row. Self-impact is shown in the diagonal panels 
        as red solid lines, cross-impact is shown as blue dashed lines and present in all panels. Price impact is calculated as average price change (multiplied by the trade sign) after a lag of 2 seconds, the minimum time that ensures we observe an update of the limit order book. Self- and cross-impact is clearly non-linear. For comparison the solid black line in the lower left panel illustrates a linear impact function. }
		\label{fig:f_measurement_all}
\end{figure}

Figure \ref{fig:f_measurement_all} shows self- and cross-impact between all bonds in our sample as a function of trade size $V$ measured in units of face value. Cross-impact is universally present across our sample and on average smaller than self-impact by roughly one order of magnitude. The cross-impact curves of different pairings $ij$ are very close one to the other when both bonds have a time-to-maturity of at least four years left. For bonds with three or less years left until maturity we do not observe an intense trading activity, thus  the curves in the leftmost column in the figure are very noisy. Likely the price-dynamics of these short-term titles are more decoupled from the medium- and long-term bonds with a lifespan of four or more years. The figure shows that all the estimated functions $f^{ij}(V)$ are non-linear, being concave and well described by a power law behavior with an exponent smaller than 1. This has been already observed in self-impact (\cite{LilloNature}) and is extended here to cross-impact.\footnote{In principle this observation suggests the presence of arbitrage opportunities due to the violation of Lemma \ref{lemma:finite_linear}. However we should remember that what is shown in Figure \ref{fig:f_measurement_all} is the observed impact, which might be different from the virtual impact, since the former does not take into account the selection bias due to the fact that traders condition the market order volume to what is present at the opposite best. For a discussion of this point in the self-impact case, see \cite{bouchaud2008markets}.}  

Having established the evidence for cross-impact, we investigate its possible origin: Is this due to correlated trades across assets (e.g. a strategy trading several bonds simultaneously) or is it mostly due to quote revision following a trade, leading to changes of the mid-price of a bond in the absence of trades? To discriminate between these alternatives, we repeat the analysis in Figure \ref{fig:f_measurement_all} and distinguish now whether there were any trades beyond the triggering one in any other bond in our sample during a period from 3 seconds before to 2 seconds after the triggering transaction, which we will call \textit{isolated trades}. 
\begin{figure}[!ht]
    \centering
        \includegraphics[width=\textwidth]{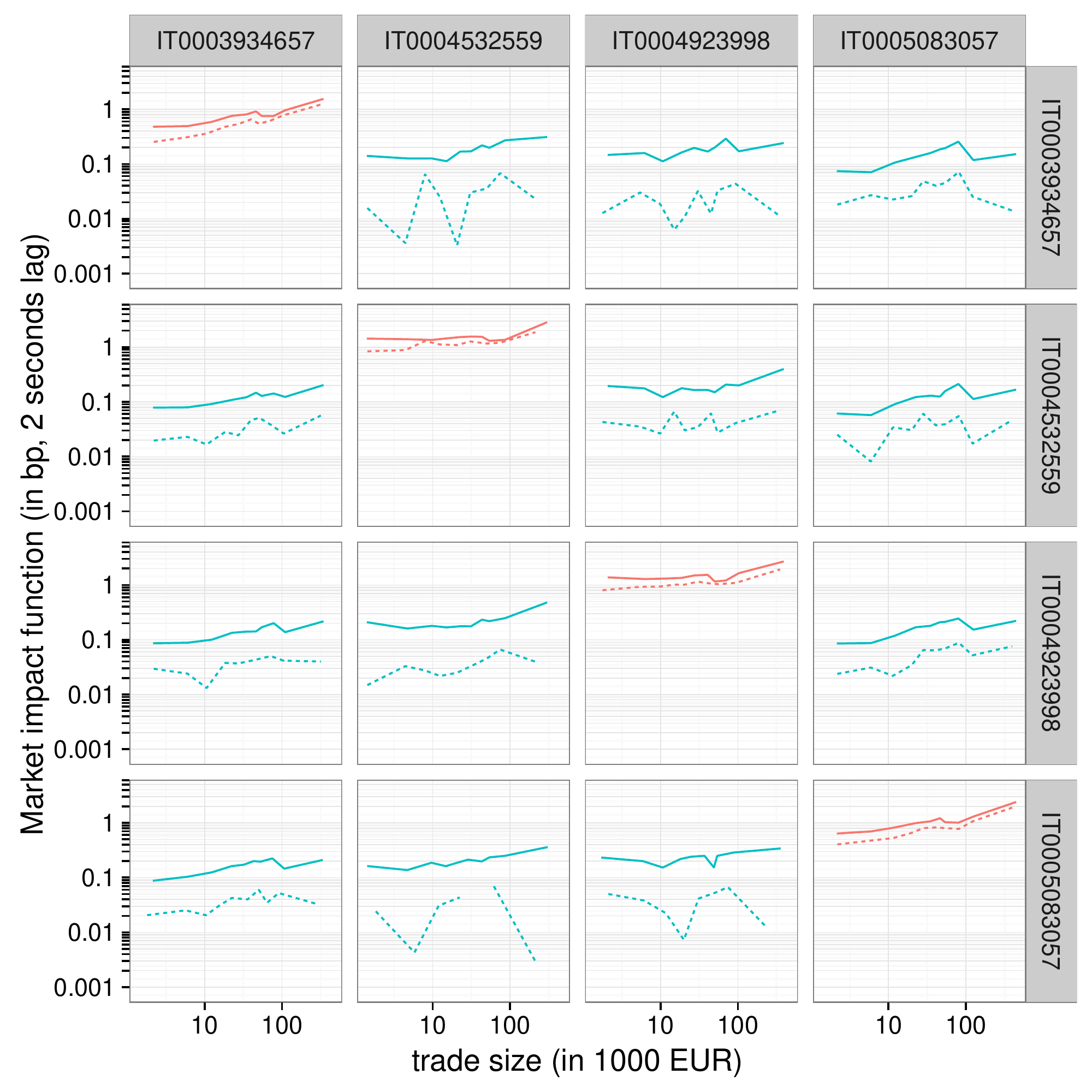}
		\caption{Plot of the average self- and cross-impact function among the four most recent 30 year bonds in our sample. Each line corresponds to one pairing $ij$, where impact is from the column on the row. Self-impact is present on the diagonal panels in red, cross-impact on the off-diagonals in blue. Solid lines show the market impact function based on all trades as in Figure \ref{fig:f_measurement_all}, dotted lines show market impact based on isolated trades only, i.e. when there was no other transaction from 3 seconds before to 2 seconds after the triggering trade. }
		\label{fig:f_measurement_sametime}
\end{figure}
For better readability in Figure \ref{fig:f_measurement_sametime} we focus on the four most recently issued 30 year BTPs in our sample, which were shown in the lower right panel of Figure \ref{fig:f_measurement_all}. Results are similar for all other pairs of bonds. 
When we consider market impact of isolated trades only, self impact is lower than unconditionally. This is somewhat expected since order signs are positively autocorrelated and we exclude contributions where other trades have on average  a positive contribution to impact. 
However the decrease in market impact is stronger for the cross-impact components, which are smaller by a factor of $\sim 5-10$ on average, whereas self-impact decreases only by a factor $\sim 2$ on average. 
We conclude therefore that both an autocorrelation of orders across assets as well as quote revisions play a role in forming cross-impact. In the next section we will take into account the (cross-) autocorrelation of the order sign when we estimate the shape of the decay of market impact.

\subsection{Decay kernel}
\label{sec:emp_decay}

To estimate the empirically observed decay function we employ a multivariate version of the transient impact model of \cite{bouchaud2004fluctuations} and similarly to \cite{benzaquen2016dissecting,wang2016microscopic}. While the advantage of the model lies in the fully non-parametric estimation of the kernel that we obtain, the TIM is typically estimated in event time which is asset-specific. Previous approaches avoid potential pitfalls by estimating the propagator in calendar time and binning trades. The estimation then is sensitive to the bin width. A small bin-width such as 1 second in \cite{wang2016microscopic} introduces problems in the treatment of bins without trading activity, while a large bin width such as 5 minutes in \cite{benzaquen2016dissecting} is too coarse to observe effects of single transactions. The main difference of our estimation is that we estimate the propagator in a combined market order time. Specifically our combined trade time is defined to advance by one unit for any unique timestamp at which there is at least one trade recorded, irrespective of the asset(s).\footnote{In other words, each trade advances time by one step, unless when there are two or more trades (in the same or different assets) recorded at exactly the same timestamp (at millisecond resolution). In such a case our combined trade time advances only by 1. In our sample ca. $3\%$ of trades happen at the same time-stamp as another trade in a different bond. }

Our model for the (log-) mid-price $X^i_t$ of asset $i$ just before a trade at time $t$ reads 
\begin{equation}
X^i_t = \sum_{t' < t}{ \left\{ \sum_j{ \left[ H^{ij}(t-t') \epsilon^j_{t'} I^j_{t'} \right] } + \xi^i_{t'} \right\} } + X^i_{-\infty}
\label{eq:mvtim_m}
\end{equation}
where $\epsilon^i_t$ is the order sign and $I^i_t$ an indicator function for a trade in asset $i$ at time $t$ as defined in Section \ref{sec:mot}. 
$\xi$ is a noise term with correlation matrix ${\bm \Sigma}^{(\xi)}$ and the empirically observed correlation structure of returns ${\bm r}$ of ${\bm X}$ is not ${\bm \Sigma}^{(\xi)}$ but the noise component ${\bm \Sigma}^{(\xi)}$ plus the component due to the correlated order flow and cross-impact ${\bm \Sigma}^{({\bm H})}$, as shown in \cite{benzaquen2016dissecting}. Finally self- and cross-impact is captured by the propagator matrix $H^{ij}(\delta t)$ which gives the price impact of a trade in asset $j$ on asset $i$ after a positive time lag $\delta t$. Note that here we assume that trades of all volumes have the same impact and to avoid confusion with the previous sections we denote the decay kernel ${\bm H}$.\footnote{${\bm H}$ corresponds to the elementwise product of ${\bm f}$ and ${\bm G}$ as defined in equation (\ref{eq:price_ND}), given the assumption of indifference to trade size.} In this model returns $r^i_t$ in asset $i$ from a trade at time $t$ to the next time-step are then defined as 
\begin{align}
	r^i_t 	&= X^i_{t+1} - X^i_t \nonumber \\
			&= \sum_j{ \underbrace{H^{ij}(1)}_{{\cal H}^{ij}(0)} \epsilon^j_{t} I^j_{t}  } + \sum_j{ \sum_{t'<t}{{\underbrace{H^{ij}(t+1-t') - H^{ij}(t-t')}_{{\cal H}^{ij}(t-t')}}^{ij}(t-t') \epsilon^j_{t'} I^j_{t'}  } } + \xi^i_t	\nonumber \\
   			&= \sum_j{ \sum_{t'\leq t}{{\cal H}^{ij}(t-t') \epsilon^j_{t'} I^j_{t'}  } }  + \xi^i_t                 
  \label{eq:mvtim_r}
\end{align}
where ${\bm{\mathcal{H}}} (\ell)\equiv {\bm H}(\ell+1)- {\bm H}(\ell)$, ${\bm H}(\ell \leq 0) \equiv 0$ and due to the definition of the price process in equation (\ref{eq:mvtim_m}) a lag of $\tau = 0$ as the argument of $\bm G$ in equation (\ref{eq:price_ND}) corresponds to $\ell = 1$ for $\bm H$. In practice (both due to computational limitations and to avoid dealing with overnight effects) the sum over $t'$ is performed up to a cutoff lag $p$. For an estimation of ${\bm H}$ that is more stable with respect to $p$ (\cite{eisler2012price}) we compute the observable $\tilde{\mathcal{S}}^{ij}(\ell)$ 
\begin{align} 
  \tilde{\mathcal{S}}^{ij}(\ell) 	
  	&= \mathbb{E}[r^i_{t+\ell} \epsilon^j_{t} I^j_{t} ] \\   
	&= \sum_k{ \sum_{n \geq 0}{ {\cal H}^{ik}(n) \mathbb{E}\left[ \epsilon^k_{t+\ell-n} I^k_{t+l-n} \epsilon^j_{t} I^j_{t} \right] } } \nonumber \\
	&= \sum_k{ \sum_{n \geq 0}{ {\cal H}^{ik}(n) {\tilde{C}^{kj}(\ell - n) }} }                                                                     
	\label{eq:mvtim_s}
\end{align}
where $\bm{\tilde{C}}(\ell - n)$ is the cross-correlation matrix of the modified order sign $\epsilon^i_{t} I^i_{t}$ at lag $\ell - n$.\footnote{Note that even though we refer to it as such, $\bm{\tilde{C}}$ is not strictly speaking a correlation matrix, as we do not de-mean nor normalize $\epsilon^i_{t} I^i_{t}$. } To estimate $H^{ij}(n) = \sum_{l=0}^{n-1}{ {\cal H}^{ij}}$ we re-write equation (\ref{eq:mvtim_s}) as a matrix equation 
\begin{equation}
\bm{\tilde{\mathcal{S}}} = \bm{{\cal H}} \bm{\tilde{\mathbf{C}}}
\label{eq:mvtim_matrix}
\end{equation}
where with a slight abuse of notation $\bm{\tilde{\mathcal{S}}}$ and $\bm{{\cal H}}$ are row vectors of $N \times N$ block matrices, i.e. $\bm{\tilde{\mathcal{S}}} = (\bm{\tilde{\mathcal{S}}}(0), \cdots, \bm{\tilde{\mathcal{S}}}(p-1))$ and $\bm{{\cal H}} = \bm{({\cal H}}(0), \cdots, \bm{{\cal H}}(p-1))$, and $\bm{\tilde{C}}$ is a symmetric block-Toeplitz matrix of $p \times p$ blocks of the correlation matrices at different lags, of dimension $N p \times N p$, \begin{equation} 
\bm{\tilde{C}} = 
 \begin{pmatrix}
  \bm{\tilde{C}}(0) & \bm{\tilde{C}}(1) & \cdots & \bm{\tilde{C}}(p-1) \\
  (\bm{\tilde{C}}(1))^{\intercal} & \bm{\tilde{C}}(0) & \cdots & \bm{\tilde{C}}(p-2) \\
  \vdots  & \vdots  & \ddots & \vdots  \\
  (\bm{\tilde{C}}(p-1))^{\intercal} & (\bm{\tilde{C}}(p-2))^{\intercal} & \cdots & \bm{\tilde{C}}(0) 
 \end{pmatrix} 
\label{eq:mvtim_cmat}
\end{equation}
where we use that $\bm{\tilde{C}}(-m)) = (\bm{\tilde{C}}(m))^{\intercal}$. To estimate $\bm{{\cal H}}$ and thus $\bm{H}$ we invert $\bm{\tilde{C}}$ and right-multiply equation (\ref{eq:mvtim_matrix}) with $\bm{\tilde{C}}^{-1}$, where both $\bm{\tilde{S}}$ and $\bm{\tilde{C}}$ are constructed from (weighted) averages over daily estimations. 
\begin{figure}[!ht]
    \centering
        \includegraphics[width=0.7 \textwidth]{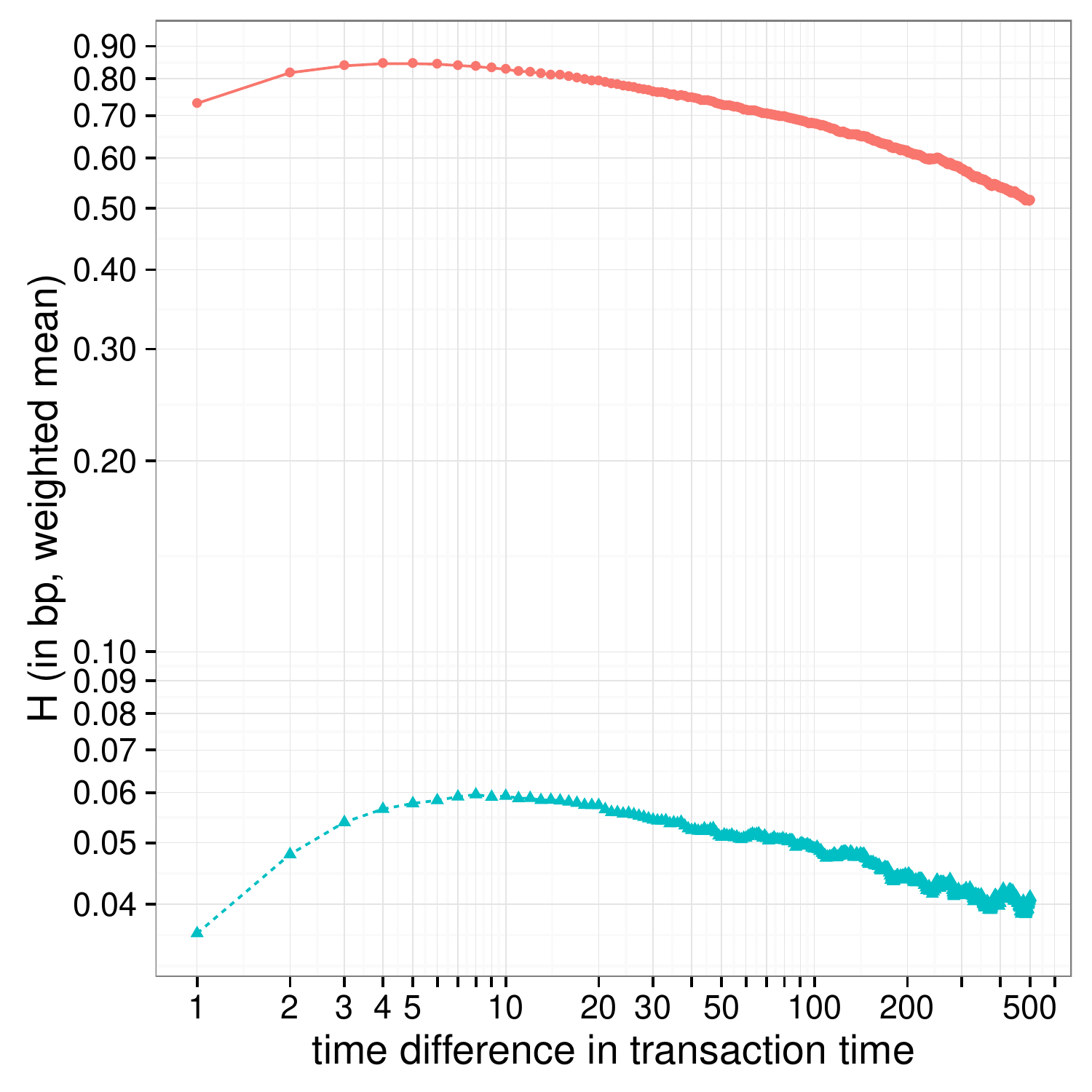}
		\caption{Plot of the estimated average decay kernel $H^{ij}$ (in basis points) for self-impact (red dots connected by solid lines) and cross-impact (blue triangles connected by dashed lines) among all bonds in our sample. For self- (cross-) impact we show the mean over all bonds (pairings) weighted by the number of transactions in the triggering bond.}
		\label{fig:G_est}
\end{figure}

Figure \ref{fig:G_est} shows the mean of the decay kernel $H^{ij}(\tau)$ for self- and cross-impact averaged over all the bonds and pairings and weighted by the number of transactions. The mean and median values are not shown here but behave similarly. Both propagators do not decay immediately but reach their peak after $\sim 10$ transactions. This indicates a market inefficiency which has been observed for self-impact in other markets (see e.g. Figure 1 in \cite{taranto2016linear}). In the absence of slippage this inefficiency could be exploited by e.g. a simple buy-hold-sell strategy. However here the expected gain is on the order of $\sim 0.1$ basis points while spread costs are $> 1$ basis points so that such a strategy would not be profitable. Further we observe that self- and cross-impact decay rather slowly with average self-impact reaching its initial level after $\sim 100$ transactions, corresponding to $\sim 10$ minutes of physical time and cross-impact taking even longer.

\subsection{Testing for symmetry of cross-impact}
\label{sec:test_symm}

We have shown in section \ref{sec:finite_symmetry} that for a bounded decay kernel the strength of cross-impact must be symmetric across pairs, i.e. $\eta^{ij} = \eta^{ji}$. Here we check whether this is empirically verified. 
In the estimation of the previous section where we are averaging over the trade volume, effectively regressing returns on trade events, this corresponds to the condition that $\hat{H}^{ij}(1) = \hat{H}^{ji}(1)$, i.e. we are assuming that prices are roughly constant so that absolute returns can be approximated by relative returns and that the average value (trade volume weighted by price) does not differ across bonds. As a robustness check, we repeat the estimation taking into account trading value, i.e. we modify Equation (\ref{eq:mvtim_r}) to be
\begin{equation}
	r^i_t = \sum_k{ \sum_{t'\leq t}{\tilde{\cal K}^{ik}(t-t') \epsilon^k_{t'} W^k_{t'} I^k_{t'} } }	
  \label{eq:mvtim_r_value}
\end{equation}
where we are now regressing returns on traded value $W^i_t = S^i_t V^i_t$ and the estimated impact and the decay kernel $\bm{K}(n) = \sum_{l=0}^{n-1}{\bm{\tilde{{\cal K}}}(l)}$ is connected to the $\bm{\eta}$ and $\bm{G}$ discussed in Sections \ref{sec:model} and \ref{sec:generalconstraints} via
\begin{equation}
	\tilde{K}^{ij}(t-s) = \tilde{\eta}^{ij} G^{ij}(t-s) = \frac{\eta^{ij}}{S^i_t S^j_t} G^{ij}(t-s)
  \label{eq:eta_tilde}
\end{equation} 
when assuming a linear $\bm{f}$. Clearly the symmetry of $\eta^{ij}$  of Lemma \ref{lemma:linbounded_symm} must hold also for $\tilde{\eta}^{ij} = \tilde{K}^{ij}(1)$ which is the impact and decay kernel estimated at its smallest lag. Again we assume a roughly constant bond price process $\bm{S}$. The added accuracy of this estimation due to including the value is countered by the fact that the empirically observed market impact function is non-linear.  

While we may easily check for symmetry on the estimated impact matrix $\bm{\eta}$, this does not allow for any statement on its statistical significance. Therefore we repeat the  estimation on a shorter time scale, i.e. we obtain $\bm{\tilde{\mathcal{S}}}$ and $\bm{\tilde{\mathbf{C}}}$ by averaging over the days of each calendar week instead of over the whole sample period and estimate the decay kernel $\bm{H}_w$ (or $\bm{K}_w$ respectively for the estimation on trade value) for each week $w$ separately. For each of the 41 estimated $\bm{H}_w$  we compute the asymmetry $\Delta H^{ij}_w = H^{ij}_w(1) - H^{ji}_w(1) $ and for each of the $33\times 32 /2=528$ pairs we perform a Student's t-test  of the null hypothesis that $\Delta H^{ij}_w = 0$. For robustness we repeat this for three different aggregation periods: weekly as described above, bi-weekly, and monthly.  

\begin{table}[ht]
\centering
\begin{tabular}{ll|rrr}
\hline 
 \multicolumn{2}{c|}{Percentage of significantly asymmetric pairs} & \multicolumn{3}{c}{ confidence level} \\
 regression on & aggregation & 1\% & 5\% & 10\% \\
\hline 
\multirow{3}{*}{trade events} & weekly & $8.0\%$ & $16.3\%$ & $24.4\%$\\
 & bi-weekly & $6.1\%$ & $15.3\%$ & $24.6\%$ \\
 & monthly & $4.0\%$ & $14.0\%$ & $24.1\%$ \\
\hline 
\multirow{3}{*}{trade value} & weekly & $3.0\%$ & $12.1\%$ & $21.8\%$ \\
 & bi-weekly & $3.4\%$ & $11.6\%$ & $21.0\%$ \\
 & monthly & $2.5\%$ & $11.0\%$ & $21.8\%$ \\
\hline 
\end{tabular}
\caption{Percentage of bond-pairs for which the null of symmetry in cross-impact is rejected according to a t-test on the null $\Delta H^{ij}_w =  H^{ij}_w(1) - H^{ji}_w(1) = 0$ ($\Delta K^{ij}_w = K^{ij}_w(1) - K^{ji}_w(1) = 0$). Tests are performed  on weekly/bi-weekly/monthly estimations of $\bm{H}$ ($\bm{K}$) from regressions of returns on signed trades (value of trades).}
\label{tab:asymm}
\end{table}
In Table \ref{tab:asymm} we report the number of pairs for which the null hypothesis that $\Delta H^{ij}=0$ ($\Delta K^{ij}=0$) is rejected. The table reveals that for all scenarios and for all confidence levels the number of bond pairs for which the assumption of symmetric cross-impact is not supported is larger than the number expected under the null hypothesis.\footnote{We have been unable to make out any obvious patterns which pairs are significantly asymmetric when ordering by various measures of liquidity and trading activity (time-to-maturity, maturity, bid-ask spread, average number of trades per day, average trade volume, turnover, tick size). This suggests that the asymmetry we observe is not just a mere artifact of any of those measures. } This implies that in principle it is possible to exploit this dynamic arbitrage opportunity in at least some pairs, for example by using the strategy presented in Section \ref{sec:finite_symmetry}. 

We now check whether such a strategy would also be easily profitable in the bond pairs singled out above when taking into account bid-ask spread costs. We proxy the slippage by the bid-ask spread $B^i$ of bond $i$ and thus for the strategy in equation (\ref{eq:strategy_asymm}) we obtain 
\begin{align}
\label{eq:slippage}
C^{\text{slippage}} &\simeq (v^a B^a + v^b B^b) \frac{T}{3} \\
C^{\text{cross}} &\simeq v^a v^b \frac{T^2}{18} \Delta \eta
\end{align}
where $\Delta \eta = |\eta^{ab} - \eta^{ba}|$. In order to make a profit the ratio
\begin{equation}
\frac{C^{\text{cross}}}{C^{\text{slippage}}} \simeq \frac{v^a v^b T \Delta \eta}{6 (v^a B^a+ v^b B^b)} \: ,
\label{eq:costeff_asymm}
\end{equation} 
which is scaling as $v T$, must be larger than one. To evaluate this ratio we need to make further assumptions on the the trading rate $v$ and the execution duration $T$. 
For the duration $T$, we need to keep in mind that in the proof of Lemma \ref{lemma:linbounded_symm} we operated in the limit of very fast trading, i.e. under the assumption that the kernel is approximately constant. The most conservative estimate for $T$ then would be 3 units of  trade time, as this is the fastest we can execute the three phases of the strategy. On the other extreme, the empirically observed decay of impact in Figure \ref{fig:G_est} suggests that the kernel actually first increases for $\sim 10$ trade time units and then decays slowly, reaching its initial value only after $\sim 100$ ($\sim 500$) trade time units for self-impact (cross-impact). Thus the maximal value of $T$ which is consistent with the constant kernel is of the order of $\sim 100$ trade time units. If we assume too high of a value for the trading rate $v$, the assumption of linear slippage costs no longer holds, as liquidity in the limit order book would be consumed faster than  being replenished and additional price impact costs would arise. We therefore suppose that an arbitrageur would use an average-sized trading rate and assume $v^i T$ as three times the average trade value in asset $i$ as reported in Table \ref{tab:bonds}.\footnote{To see this, denote the average trade size in asset $i$ as $\bar{x}^i$ shares and take the case of executing the strategy in three trades corresponding to 3 units of trade time. The first phase, i.e. the first trade, lasts $T/3$ and is of size $\bar{x}^i$ shares, therefore $v^i T = 3 \bar{x}^i$ shares.}

In the following we estimate the ratio in equation (\ref{eq:costeff_asymm}) for the set of 6 pairs ($1.1\%$) where symmetry is rejected at the $5\%$ level for at least five of the six aggregation and regression scenarios of Table \ref{tab:asymm}.\footnote{Considering different sets leads to similar results.} Conservatively assuming $T$ as $3$ units of trade time and with the assumptions described above, we get $\frac{C^{\text{cross}}}{C^{\text{slippage}}} \sim 1 \cdot 10^{-4}$. If we are able to maintain this strategy for a longer period $T$ (at the same trading rate), we are getting closer to profitability since gains from cross-impact scale as $T^2$ and losses due to slippage as $T$. Assuming $T = 100$, i.e. the timescale when on average impact has decayed beyond its initial timescale, yields a ratio $\frac{C^{\text{cross}}}{C^{\text{slippage}}} \sim 0.005$. Only if we were able to neglect decay and other costs when executing our strategy throughout a whole trading day, it could turn profitable. Given that we do observe a faster decay (but also considering the associated risk and the chance that dominating the trading activity with the strategy might produce a less favorable impact structure) we conclude that dynamic arbitrage from cross-impact is unprofitable at least with our simple trading strategy. However our results also indicate that it is worth taking cross-impact into account when executing other strategies.

\FloatBarrier
\section{Conclusion}
\label{sec:conclusion}

Even though cross-impact has been studied in the theoretical literature on optimal portfolio liquidation, empirical studies have been scarce until very recently. In this paper we aim to connect the two strands of literature from a no-dynamic-arbitrage perspective. 

A desirable market impact model should be free of arbitrage opportunities. In this paper we focus on the specific class of multi-asset Transient Impact Models (TIMs)  of market impact and we derive some necessary conditions for the absence of dynamic arbitrage. In particular, by using specific examples of simple round-trip strategies, we focus our attention on possible   constraints on the shape and size of cross-impact. 

One such condition is symmetry of cross-impact with respect to its direction between assets and we test it on empirical cross-impact kernels obtained by estimating a  TIM in transaction time on Italian sovereign bonds traded in the MOT electronic market. Due to the strongly interrelated nature of these assets,  we believe cross-impact plays a much stronger role here than compared to other asset classes such as, for example, equities. We find that while there exist statistically significant violations of the no-arbitrage conditions related to impact symmetry, these are unprofitable because of slippage costs such as the bid-ask spread which are neglected in the theoretical considerations. 

In addition to this, we want to stress our contributions in describing the high-frequency market microstructure of the MOT sovereign bond market, applying the TIM to fixed-income markets and presenting evidence for cross-impact at the level of single orders instead of aggregated order flows. While this type of modeling and empirical estimation has been performed on many different types of markets, the application to sovereign bond (electronic) markets is new. 
This makes our paper also relevant from a monetary policy point of view. Recent studies (\cite{schlepper2017scarcity,arrata2017price,desantis2017flow}) that aim to quantify the price impact (measured as decline in the yield to maturity) of Quantitative Easing purchases in the Euro area could benefit from taking into account cross-impact effects.

\clearpage
{
\bibliographystyle{abbrvnat}
\bibliography{cross_impact}

\begin{thebibliography}{36}
\providecommand{\natexlab}[1]{#1}
\providecommand{\url}[1]{\texttt{#1}}
\expandafter\ifx\csname urlstyle\endcsname\relax
  \providecommand{\doi}[1]{doi: #1}\else
  \providecommand{\doi}{doi: \begingroup \urlstyle{rm}\Url}\fi

\bibitem[Alfonsi et~al.(2012)Alfonsi, Schied, and Slynko]{alfonsi2012order}
A.~Alfonsi, A.~Schied, and A.~Slynko.
\newblock Order book resilience, price manipulation, and the positive portfolio
  problem.
\newblock \emph{SIAM Journal on Financial Mathematics}, 3\penalty0
  (1):\penalty0 511--533, 2012.

\bibitem[Alfonsi et~al.(2016)Alfonsi, Kl{\"o}ck, and
  Schied]{alfonsi2016multivariate}
A.~Alfonsi, F.~Kl{\"o}ck, and A.~Schied.
\newblock Multivariate transient price impact and matrix-valued positive
  definite functions.
\newblock \emph{Mathematics of {O}perations {R}esearch}, 41\penalty0
  (3):\penalty0 914--934, 2016.

\bibitem[Almgren and Chriss(2001)]{almgren2001optimal}
R.~Almgren and N.~Chriss.
\newblock Optimal execution of portfolio transactions.
\newblock \emph{Journal of Risk}, 3:\penalty0 5--40, 2001.

\bibitem[Almgren(2003)]{almgren2003optimal}
R.~F. Almgren.
\newblock Optimal execution with nonlinear impact functions and
  trading-enhanced risk.
\newblock \emph{Applied Mathematical Finance}, 10\penalty0 (1):\penalty0 1--18,
  2003.

\bibitem[Arrata and Nguyen(2017)]{arrata2017price}
W.~Arrata and B.~Nguyen.
\newblock Price impact of bond supply shocks: Evidence from the eurosystem's
  asset purchase program.
\newblock \emph{Banque de France Working Paper No. 623}, 2017.

\bibitem[Benzaquen et~al.(2017)Benzaquen, Mastromatteo, Eisler, and
  Bouchaud]{benzaquen2016dissecting}
M.~Benzaquen, I.~Mastromatteo, Z.~Eisler, and J.-P. Bouchaud.
\newblock Dissecting cross-impact on stock markets: An empirical analysis.
\newblock \emph{Journal of Statistical Mechanics: Theory and Experiment},
  2017\penalty0 (2):\penalty0 023406, 2017.

\bibitem[Bouchaud et~al.(2004)Bouchaud, Gefen, Potters, and
  Wyart]{bouchaud2004fluctuations}
J.-P. Bouchaud, Y.~Gefen, M.~Potters, and M.~Wyart.
\newblock Fluctuations and response in financial markets: the subtle nature of
  \enquote{random} price changes.
\newblock \emph{Quantitative Finance}, 4\penalty0 (2):\penalty0 176--190, 2004.

\bibitem[Bouchaud et~al.(2008)Bouchaud, Farmer, and Lillo]{bouchaud2008markets}
J.-P. Bouchaud, J.~D. Farmer, and F.~Lillo.
\newblock How markets slowly digest changes in supply and demand.
\newblock In \emph{Handbook of Financial Markets: Dynamics and Evolution}.
  Elsevier: Academic Press, 2008.

\bibitem[Busseti and Lillo(2012)]{busseti2011}
E.~Busseti and F.~Lillo.
\newblock Calibration of optimal execution of financial transactions in the
  presence of transient market impact.
\newblock \emph{Journal of Statistical Mechanics: Theory and Experiments},
  2012\penalty0 (09):\penalty0 P09010, 2012.

\bibitem[Curato et~al.(2016)Curato, Gatheral, and Lillo]{curato2016discrete}
G.~Curato, J.~Gatheral, and F.~Lillo.
\newblock Discrete homotopy analysis for optimal trading execution with
  nonlinear transient market impact.
\newblock \emph{Communications in Nonlinear Science and Numerical Simulation},
  39:\penalty0 332--342, 2016.

\bibitem[Curato et~al.(2017)Curato, Gatheral, and Lillo]{curato2016optimal}
G.~Curato, J.~Gatheral, and F.~Lillo.
\newblock Optimal execution with non-linear transient market impact.
\newblock \emph{Quantitative Finance}, 17\penalty0 (1):\penalty0 41--54, 2017.

\bibitem[Darbha and Dufour(2013)]{darbha2013microstructure}
M.~Darbha and A.~Dufour.
\newblock Microstructure of the {Euro}-area government bond market.
\newblock In H.~K. Baker and H.~Kiymaz, editors, \emph{Market Microstructure in
  Emerging and Developed Markets}, chapter~3, pages 39--58. Wiley Online
  Library, 2013.

\bibitem[De~Santis and Holm-Hadulla(2017)]{desantis2017flow}
R.~A. De~Santis and F.~Holm-Hadulla.
\newblock Flow effects of central bank asset purchases on euro area sovereign
  bond yields: evidence from a natural experiment.
\newblock \emph{ECB Working Paper}, 2017.

\bibitem[Dufour and Nguyen(2012)]{dufour2012permanent}
A.~Dufour and M.~Nguyen.
\newblock Permanent trading impacts and bond yields.
\newblock \emph{The European Journal of Finance}, 18\penalty0 (9):\penalty0
  841--864, 2012.

\bibitem[Eisler et~al.(2012)Eisler, Bouchaud, and Kockelkoren]{eisler2012price}
Z.~Eisler, J.-P. Bouchaud, and J.~Kockelkoren.
\newblock The price impact of order book events: market orders, limit orders
  and cancellations.
\newblock \emph{Quantitative Finance}, 12\penalty0 (9):\penalty0 1395--1419,
  2012.

\bibitem[Gatheral(2010)]{gatheral2010no}
J.~Gatheral.
\newblock No-dynamic-arbitrage and market impact.
\newblock \emph{Quantitative Finance}, 10\penalty0 (7):\penalty0 749--759,
  2010.

\bibitem[Gatheral and Schied(2013)]{GatheralHandbook}
J.~Gatheral and A.~Schied.
\newblock Dynamical models of market impact and algorithms for order execution.
\newblock In \emph{Handbook of Systemic Risk}, pages 579--599. Cambridge
  University Press, 2013.

\bibitem[Gatheral et~al.(2011)Gatheral, Schied, and
  Slynko]{gatheral2011exponential}
J.~Gatheral, A.~Schied, and A.~Slynko.
\newblock Exponential resilience and decay of market impact.
\newblock In \emph{Econophysics of Order-driven Markets}, pages 225--236.
  Springer, 2011.

\bibitem[Hasbrouck and Seppi(2001)]{hasbrouck2001common}
J.~Hasbrouck and D.~J. Seppi.
\newblock Common factors in prices, order flows, and liquidity.
\newblock \emph{Journal of Financial Economics}, 59\penalty0 (3):\penalty0
  383--411, 2001.

\bibitem[Huberman and Stanzl(2004)]{huberman2004price}
G.~Huberman and W.~Stanzl.
\newblock Price manipulation and quasi-arbitrage.
\newblock \emph{Econometrica}, 72\penalty0 (4):\penalty0 1247--1275, 2004.

\bibitem[Kratz and Sch{\"o}neborn(2015)]{kratz2015portfolio}
P.~Kratz and T.~Sch{\"o}neborn.
\newblock Portfolio liquidation in dark pools in continuous time.
\newblock \emph{Mathematical Finance}, 25\penalty0 (3):\penalty0 496--544,
  2015.

\bibitem[Lee and Ready(1991)]{lee1991inferring}
C.~Lee and M.~J. Ready.
\newblock Inferring trade direction from intraday data.
\newblock \emph{The Journal of Finance}, 46\penalty0 (2):\penalty0 733--746,
  1991.

\bibitem[Lillo et~al.(2003)Lillo, Farmer, and Mantegna]{LilloNature}
F.~Lillo, J.~D. Farmer, and R.~N. Mantegna.
\newblock Master curve for price-impact function.
\newblock \emph{Nature}, 421:\penalty0 129--130, 2003.

\bibitem[Linciano et~al.(2014)Linciano, Fancello, Gentile, and
  Modena]{linciano2014liquidity}
N.~Linciano, F.~Fancello, M.~Gentile, and M.~Modena.
\newblock The liquidity of dual-listed corporate bonds. {Empirical} evidence
  from {Italian} markets.
\newblock \emph{CONSOB Working Paper}, 2014.

\bibitem[Mastromatteo et~al.(2017)Mastromatteo, Benzaquen, Eisler, and
  Bouchaud]{mastromatteo2017trading}
I.~Mastromatteo, M.~Benzaquen, Z.~Eisler, and J.-P. Bouchaud.
\newblock Trading lightly: Cross-impact and optimal portfolio execution.
\newblock \emph{Available at SSRN 2949748}, 2017.

\bibitem[Obizhaeva and Wang(2013)]{obizhaeva2013optimal}
A.~A. Obizhaeva and J.~Wang.
\newblock Optimal trading strategy and supply/demand dynamics.
\newblock \emph{Journal of Financial Markets}, 16\penalty0 (1):\penalty0 1--32,
  2013.

\bibitem[Pasquariello and Vega(2013)]{pasquariello2013strategic}
P.~Pasquariello and C.~Vega.
\newblock Strategic cross-trading in the {U.S.} {s}tock market.
\newblock \emph{Review of Finance}, 19:\penalty0 229--282, 2013.

\bibitem[Schied et~al.(2010)Schied, Sch{\"o}neborn, and
  Tehranchi]{schied2010optimal}
A.~Schied, T.~Sch{\"o}neborn, and M.~Tehranchi.
\newblock Optimal basket liquidation for {CARA} investors is deterministic.
\newblock \emph{Applied Mathematical Finance}, 17\penalty0 (6):\penalty0
  471--489, 2010.

\bibitem[Schlepper et~al.(2017)Schlepper, Hofer, Riordan, and
  Schrimpf]{schlepper2017scarcity}
K.~Schlepper, H.~Hofer, R.~Riordan, and A.~Schrimpf.
\newblock Scarcity effects of {QE}: A transaction-level analysis in the bund
  market.
\newblock \emph{Bundesbank Discussion Papers}, 2017.

\bibitem[Schneider et~al.(2016)Schneider, Lillo, and
  Pelizzon]{schneider2016has}
M.~Schneider, F.~Lillo, and L.~Pelizzon.
\newblock How has sovereign bond market liquidity changed? {A}n illiquidity
  spillover analysis.
\newblock \emph{SAFE Working Paper Series}, 2016.

\bibitem[Sch{\"o}neborn(2016)]{schoneborn2016adaptive}
T.~Sch{\"o}neborn.
\newblock Adaptive basket liquidation.
\newblock \emph{Finance and Stochastics}, 20\penalty0 (2):\penalty0 455--493,
  2016.

\bibitem[Taranto et~al.(2016)Taranto, Bormetti, Bouchaud, Lillo, and
  Toth]{taranto2016linear}
D.~E. Taranto, G.~Bormetti, J.-P. Bouchaud, F.~Lillo, and B.~Toth.
\newblock Linear models for the impact of order flow on prices {I}.
  {P}ropagators: Transient vs. history dependent impact.
\newblock \emph{arXiv:1602.02735}, 2016.

\bibitem[Tsoukalas et~al.(2017)Tsoukalas, Wang, and
  Giesecke]{tsoukalas2016dynamic}
G.~Tsoukalas, J.~Wang, and K.~Giesecke.
\newblock Dynamic portfolio execution.
\newblock \emph{Management Science, forthcoming}, 2017.

\bibitem[Wang and Guhr(2016)]{wang2016microscopic}
S.~Wang and T.~Guhr.
\newblock Microscopic understanding of cross-responses between stocks: a
  two-component price impact model.
\newblock \emph{arXiv:1609.02395}, 2016.

\bibitem[Wang et~al.(2016{\natexlab{a}})Wang, Sch{\"a}fer, and
  Guhr]{wang2016average}
S.~Wang, R.~Sch{\"a}fer, and T.~Guhr.
\newblock Average cross-responses in correlated financial markets.
\newblock \emph{The European Physical Journal B}, 89\penalty0 (9):\penalty0
  207, 2016{\natexlab{a}}.

\bibitem[Wang et~al.(2016{\natexlab{b}})Wang, Sch{\"a}fer, and
  Guhr]{wang2016cross}
S.~Wang, R.~Sch{\"a}fer, and T.~Guhr.
\newblock Cross-response in correlated financial markets: individual stocks.
\newblock \emph{The European Physical Journal B}, 89\penalty0 (4):\penalty0
  1--16, 2016{\natexlab{b}}.

\end{thebibliography}
}

\newpage
\appendix


\section{Proofs}
\label{app:proofs}

Let us recall that for all proofs in this section we assume that the decay kernel $\bm{G}$ is bounded, i.e. there exists an upper bound $U > 0$ so that $|G^{ij}(\tau)| < U$ for all $\tau \in [0,\infty)$ and all $i,j$ and therefore we take $\bm{G}$ as normalized to $1$ for its smallest lag $\tau = t-s$, i.e. $G^{ij}(0) = 1$ for all pairs $ij$. Further we assume $\bm{G}(\tau)$ to be non-increasing and right-continuous at $\tau = 0$, i.e.
\begin{equation}
\label{eq:def_continuity}
\forall \: \varepsilon > 0 \: \exists \: T_{\varepsilon} > 0 \: \text{such that} \: \forall \: \tau \: \text{with} \: 0 < \tau < T_{\varepsilon} \: \text{and} \: \forall \: i,j: \: |G^{ij}(0) - G^{ij}(\tau)| < \varepsilon \: .
\end{equation}

\begin{proof}[Proof of Lemma \ref{lemma:antisym_ND}]
We first show that a non-odd $f$ leads to a price manipulation in the single-asset case. Let us therefore assume an in-out strategy $\Pi$ as in the first component of Example \ref{eg:strategy_easy} with $\kappa = -1$ and therefore $\vartheta = 1/2$, i.e. both phases of trading last equally long. That is we first accumulate a position at the rate $v > 0$ to then liquidate it at the same negative rate $-v$. 
Without loss of generality we choose $v$ such that $f(v) > -f(-v) \geq 0$ and that 
\begin{equation}
\label{eq:eps_odd}
\varepsilon \coloneqq  \frac{1}{4} \frac{f(v) + f(-v)}{f(v)} > 0
\end{equation}
which is non-zero since $f$ is not an odd function.\footnote{In the case that $f(v) < -f(-v)$ all we need is to change the sign in equation (\ref{eq:eps_odd}) to ensure that $\varepsilon$  is positive. In the cases that either $f(v) \leq 0 \wedge f(-v) \leq 0$ or $f(v) \geq 0 \wedge f(-v) \geq 0$ the price manipulation arises in a simple in-out or out-in strategy as above respectively. Finally if assuming $vf(v) \leq 0$ the proof is analogous to the one above.} Then by continuity there exists a $T_{\varepsilon} > 0$ for which we can bound the cost of our strategy as 
\begin{align}
C(\Pi) &= v f(v) \int_0^{T/2}{\mathrm{d}t \int_0^t{G(t-s) \mathrm{d}s } } \nonumber \\
			&\qquad -v f(-v) \int_{T/2}^{T}{\mathrm{d}t \int_{T/2}^t{G(t-s) \mathrm{d}s } } \nonumber \\ 
            & \qquad -v f(v) \int_{T/2}^{T}{\mathrm{d}t \int_{0}^{T/2}{G(t-s) \mathrm{d}s } } \nonumber \\
       &\leq v f(v) \frac{T^2}{8} - v f(-v) \frac{T^2}{8} - v f(v) (1-\varepsilon) \frac{T^2}{4} \nonumber \\ 
       &= v \frac{T^2}{8} \left[ -f(v) - f(-v) + 2\varepsilon f(v) \right] \nonumber \\
       &< 0
\end{align}
when choosing $T = T_{\varepsilon}$. That is for any trading rate $v$ there is a $ T_{\varepsilon}>0$ for which there is a price manipulation by our choice of $\varepsilon$. Therefore we conclude that $f(v)$ must be an odd function of $v$ in the single-asset case and the same holds for self-impact $f^{ii}$ in the multi-asset case with cross-impact since we can always execute a strategy in only one asset. 

Let us now show that the same holds for cross-impact. We choose $v_a, v_b > 0$ and for simplicity we assume that $\sgn (\dot{x}_i f^{ij}(\dot{x}_j) = \sgn (\dot{x}_i \dot{x}_j)$ for all $i,j$, the proof being analogous in the other cases. We re-define 
\begin{equation}
\label{eq:eps_odd_2d}
\varepsilon \coloneqq  \frac{v_a\left[ f^{ab}(v_b) + f^{ab}(-v_b) \right]}{4 v_a f^{aa}(v_a) + 4 v_b f^{bb}(v_b) + 3 v_a f^{ab}(v_b) - v_a f^{ab}(-v_b) + 3 v_b f^{ba}(v_a) - v_b f^{ba}(-v_a) } > 0
\end{equation}
where again we assume for simplicity that $f^{ab}(v_b) + f^{ab}(-v_b) > 0$, the proof being analogous in the other case with $f^{ab}(v_b)$ and $f^{ab}(-v_b)$ interchanged in (\ref{eq:eps_odd_2d}). 
We assume a first strategy $\Pi_1$ as in Example \ref{eg:strategy_easy}, again with $\kappa = -1$ and therefore $\vartheta = 1/2$, and $\lambda > 0$, that is we first accumulate both assets at the rates $v_a, v_b >0$ to then liquidate both positions at the negative rates $-v_a, -v_b < 0$:  
\begin{equation}
\Pi_1 = \left\{ \dot{\bm{x}}_t \right\} \quad, \quad \dot{\bm{x}}_t = \begin{cases} 
\left(+v_a, +v_b \right)^{\intercal} \quad \text{for} \quad 0 \leq t \leq T/2 \\
\left(-v_a, -v_b \right)^{\intercal} \quad \text{for} \quad T/2 < t \leq T 
\end{cases}
\label{eq:trading_perm_1}
\end{equation}
Choosing $T = T_{\varepsilon}$ we are able to bound the cost of this strategy as above, obtaining
\begin{align}
C(\Pi_1)	&\leq \frac{T^2}{8} \sum_{i,j = a,b}{ v_i f^{ij}(v_j) -v_if^{ij}(-v_j) - 2v_i (1-\varepsilon) f^{ij}(v_j) } \nonumber \\
					&= 		\frac{T^2}{8} \left\{ -v_a \left[ f^{ab}(v_b) + f^{ab}(-v_b)  \right] - v_b \left[ f^{ba}(v_a) + f^{ba}(-v_a) \right] + 2 \varepsilon \sum_{i,j = a,b}{v_i f^{ij}(v_j)} \right\} \: .
\label{eq:cost_ND_Pi1}
\end{align}
Repeating the same estimation for a strategy $\Pi_2$ which is anti-symmetric with $\lambda < 0$ but otherwise as above, i.e. 
\begin{equation}
\Pi_2 = \left\{ \dot{\bm{x}}_t \right\} \quad, \quad \dot{\bm{x}}_t = \begin{cases} 
\left(+v_a, -v_b \right)^{\intercal} \quad \text{for} \quad 0 \leq t \leq T/2 \\
\left(-v_a, +v_b \right)^{\intercal} \quad \text{for} \quad T/2 < t \leq T 
\end{cases}.
\label{eq:trading_perm_2}
\end{equation}
The cost is similarly bounded from above:
\begin{align}
C(\Pi_2) 	\leq 	\frac{T^2}{8} &\left\{ - v_a \left[ f^{ab}(v_b) + f^{ab}(-v_b)  \right] + v_b \left[ f^{ba}(v_a) + f^{ba}(-v_a) \right] \right. \nonumber \\
					 &\left. + \varepsilon \left[ 2 v_a f^{aa}(v_a) + 2 v_b f^{bb}(v_b) + v_a f^{ab}(v_b) - v_a f^{ab}(-v_b) + v_b f^{ba}(v_a) - v_b f^{ba}(-v_a) \right] \right\} \: .
\label{eq:cost_ND_Pi2}
\end{align}
Combining the cost of the two strategies with identical parameters $v_a$, $v_b$ we have 
\begin{alignat}{3}
C(\Pi_1) + C(\Pi_2) 	\leq 	\frac{T^2}{4}  \left\{ \right. &- v_a  &&\left[ f^{ab}(v_b) + f^{ab}(-v_b) \right]  \nonumber \\
					 &+ \varepsilon &&\left[ 4 v_a f^{aa}(v_a) + 4 v_b f^{bb}(v_b) \right. \nonumber \\
                    &   &&\left. \left. +3 v_a f^{ab}(v_b) - v_a f^{ab}(-v_b) + 3 v_b f^{ba}(v_a) - v_b f^{ba}(-v_a) \right] \right\} 
\label{eq:cost_ND_Pi1and2}
\end{alignat}
which is negative for our choice of $\varepsilon$ and therefore price manipulation is possible. 

\end{proof}

\begin{proof}[Proof of Lemma \ref{lemma:finite_linear}]
Consider a scenario as in Example \ref{eg:strategy_easy}, i.e. trading in two assets over two phases, denoted by $\RN{1}$ and $\RN{2}$, at a constant rate during each phase, with $\lambda = v_{a,\RN{1}} / v_{b,\RN{1}} = v_{a,\RN{2}} / v_{b,\RN{2}}$, $\kappa = v_{i, \RN{1}} / v_{i, \RN{2}} = v_{\RN{1}} / v_{\RN{2}} < 0$ and the turn-around point $\Theta = T \frac{-v_{i,\RN{2}}}{v_{i,\RN{1}} - v_{i,\RN{2}} } = T \frac{1}{1 - \kappa}$ that is common to both assets. From the single-asset case we use the result that self-impact is linear and denote this as 
\begin{equation}
f^{ii}(v) = \eta^{ii} v 
\label{eq:self_linear}
\end{equation}  
and for simplicity we assume $\lambda > 0$, i.e. trading in the same direction. We define 
\begin{equation}
\label{eq:eps_linear_2d}
\varepsilon \coloneqq  - \frac{1}{4} \frac{v_{a,\RN{1}} f^{ab}(v_{b, \RN{2}}) - v_{a,\RN{2}} f^{ab}(v_{b, \RN{1}}) + v_{b,\RN{1}} f^{ba}(v_{a, \RN{2}}) - v_{b,\RN{2}} f^{ba}(v_{a, \RN{1}})}{\eta^{aa} v_{a,\RN{1}} v_{a,\RN{2}} + \eta^{bb} v_{b,\RN{1}} v_{b,\RN{2}} + v_{a,\RN{2}} f^{ab}(v_{b, \RN{1}}) + v_{b,\RN{2}} f^{ba}(v_{a, \RN{1}})} 
\end{equation}
where we are assuming that $\varepsilon > 0$, otherwise we simply replace $\kappa$ with $1/\kappa$. By continuity of $\bm{G}$ as defined in Equation (\ref{eq:def_continuity}) we can choose $T = T_{\varepsilon}$ such that cost can be bounded from above as
\begin{align}
C(\Pi) 	&= \sum_{i,j = a,b}{ C_A^{ij} + C_B^{ij} + C_C^{ij}} \nonumber \\
		&\leq \sum_{i,j = a,b}{ v_{i,\RN{1}} f^{ij}(v_{j,\RN{1}}) \frac{\Theta^2}{2} 
           					  + v_{i,\RN{2}} f^{ij}(v_{j,\RN{2}}) \frac{(T-\Theta)^2}{2}
                              + v_{i,\RN{2}} f^{ij}(v_{j,\RN{1}}) (1-\varepsilon) \Theta (T-\Theta)} \nonumber \\
        &= \frac{T^2}{2} \frac{\kappa}{(\kappa -1)^2} \left\{ v_{a,\RN{1}} f^{ab}(v_{b, \RN{2}}) - v_{a,\RN{2}} f^{ab}(v_{b, \RN{1}}) + v_{b,\RN{1}} f^{ba}(v_{a, \RN{2}}) - v_{b,\RN{2}} f^{ba}(v_{a, \RN{1}}) \right. \nonumber \\
        &\left. \qquad \qquad \qquad + 2\varepsilon \left[ \eta^{aa} v_{a,\RN{1}} v_{a,\RN{2}} + \eta^{bb} v_{b,\RN{1}} v_{b,\RN{2}} + v_{a,\RN{2}} f^{ab}(v_{b, \RN{1}}) + v_{b,\RN{2}} f^{ba}(v_{a, \RN{1}}) \right] \right\} \nonumber \\
                &= \frac{T^2}{4} \frac{\kappa}{(\kappa -1)^2} \left\{ v_{a,\RN{1}} f^{ab}(v_{b, \RN{2}}) - v_{a,\RN{2}} f^{ab}(v_{b, \RN{1}}) + v_{b,\RN{1}} f^{ba}(v_{a, \RN{2}}) - v_{b,\RN{2}} f^{ba}(v_{a, \RN{1}}) \right\} \nonumber \\
        &= \frac{T^2}{4} \frac{\kappa}{(\kappa -1)^2} \left\{ v_{\RN{1}} \left[ \lambda f^{ab}(v_{\RN{2}}) + f^{ba}( \lambda v_{\RN{2}}) \right] - v_{\RN{2}} \left[ \lambda f^{ab}(v_{\RN{1}}) + f^{ba}( \lambda v_{\RN{1}}) \right] \right\} 
\label{eq:cost_linear}
\end{align}
where in the last step we have substituted $v_{b,\RN{1}} = v_{\RN{1}}$, $v_{b,\RN{2}} = v_{\RN{2}}$, $v_{a,\RN{1}} = \lambda v_{\RN{1}}$ and $v_{a,\RN{2}} = \lambda v_{\RN{2}}$. By our choice of $\varepsilon$ this cost is negative unless the auxiliary function $\alpha (v) = \lambda f^{ab}(v) + f^{ba}(\lambda v)$ is linear in $v$, i.e. $\alpha (v) = \eta^{\alpha} v$. Since this has to hold for any $\lambda \neq 0$, we find that necessarily also $f^{ab}$ and $f^{ba}$ need to be linear in the rate of trading. 
To prove this last claim we expand $\alpha(v)$ around $v=0$ and use Corollary \ref{cor:absence} that $f^{ab}(0) = f^{ba}(0) = 0$:
\begin{align}
\nonumber
\label{eq:alpha}
\alpha(v) = \eta^{\alpha} v &= \lambda f^{ab}(v) + f^{ba}(\lambda v) \\ 
							&= \sum_{l=1}^{\infty}{ ( \lambda k_l^{ab} + \lambda^l k_l^{ba} ) v^l }
\end{align}
This can be linear for any $\lambda$ only if $k_l^{ab} = k_l^{ba} = 0$ for all $l > 1$, i.e. if both cross-impact terms $f^{ab}(v)$ and $f^{ba}(v)$ are linear in $v$. 

\end{proof}

\begin{proof}[Proof of Lemma \ref{lemma:linbounded_symm}]
Given the parameters $\eta^{ij}$ determining the strength of self- and cross-impact and for any $v_a, v_b > 0$, assuming w.o.l.g. that $\eta^{ab} > \eta^{ba} > 0$, we choose 
\begin{equation}
\label{eq:eps_symm_2d}
\varepsilon \coloneqq \frac{1}{2} \frac{v_a v_b \left( \eta^{ab} - \eta^{ba} \right)}{2 \eta^{aa} v_a^2 + 2 \eta^{bb} v_b^2 + 3 \eta^{ab} v_a v_b + \eta^{ba} v_a v_b}
\end{equation}
so that $\varepsilon > 0$.\footnote{We can choose an equivalent $\varepsilon > 0$ in all other cases, i.e. for $\eta^{ba} > \eta^{ab} > 0$ by interchanging $a \leftrightarrow b$ in Equation (\ref{eq:eps_symm_2d}) and below. In the cases that either one or both of $\eta^{ab}$ and $\eta^{ba}$ are less than zero, $v_a$ and $v_b$ are to be chosen such that the denominator in (\ref{eq:eps_symm_2d}) is larger than zero and if the numerator in (\ref{eq:eps_symm_2d}) is less than zero we interchange $a \leftrightarrow b$ in order to ensure $\varepsilon > 0$.} 
Then by our assumption of continuity there exists a $T_{\varepsilon} > 0$ for which $|G^{ij}(0) - G^{ij}(\tau)| \leq \varepsilon$ for all $\tau$ with $0 \leq \tau \leq T_{\varepsilon}$ and for all $ i,j \in \{a,b\}$.
We implement the same asymmetric strategy (\ref{eq:strategy_asymm}) as in Example \ref{eg:linperm_asymm} with $T = T_{\varepsilon}$ and calculate the cost of this strategy as $C(\Pi) = C^{aa} + C^{bb} + C^{ab} + C^{ba}$ where $C^{aa}$ and $C^{bb}$ are the self-impact costs of trading in assets $a$ and $b$ respectively and $C^{ab}$ and $C^{ba}$ are the costs due to cross-impact from asset $b$ to $a$ and vice versa. The explicit calculation of the self-impact terms gives
\begin{align}
C^{aa} &= \eta^{aa} v_a^2 \left\{ \int_0^{T/3} \mathrm{d}t \int_0^t G^{aa}(t-s) \mathrm{d}s + \int_{2T/3}^{T} \mathrm{d}t \int_{2T/3}^t G^{aa}(t-s) \mathrm{d}s - \int_{2T/3}^{T} \mathrm{d}t \int_0^{T/3} G^{aa}(t-s) \mathrm{d}s \right\}  	\nonumber 	\\ 
	   &\leq   \eta^{aa} v_a^2 \left\{ \frac{T^2}{18} + \frac{T^2}{18} - (1-\varepsilon) \frac{T^2}{9} \right\} \nonumber \\
       &= \varepsilon \eta^{aa} v_a^2 \frac{T^2}{9} \nonumber \\
C^{bb} &\leq \varepsilon \eta^{bb} v_b^2 \frac{T^2}{9} \label{eq:cost_asymm_self}
\end{align}
and likewise for cross-impact:
\begin{align}
C^{ab} &= \eta^{ab} v_a v_b \left\{ \!-\! \int_0^{\!T/3} \mathrm{d}t \int_0^t \!G^{ab}(t-s) \mathrm{d}s +\! \int_{2T/3}^{T}\! \mathrm{d}t \int_{0}^{T/3}\! G^{ab}(t-s) \mathrm{d}s -\! \int_{2T/3}^{T}\! \mathrm{d}t \int_{T/3}^{2T/3} \!G^{ab}(t-s) \mathrm{d}s \right\}	 \nonumber \\
	   &\leq \eta^{ab} v_a v_b \left\{ -(1-\varepsilon) \frac{T^2}{18} + \varepsilon \frac{T^2}{9} \right\} \nonumber \\
       &= \eta^{ab} v_a v_b \frac{T^2}{18} (-1 + 3\varepsilon) \nonumber \\
C^{ba} &= \eta^{ba} v_a v_b \left\{ - \int_0^{T/3} \mathrm{d}t \int_0^t G^{ba}(t-s) \mathrm{d}s + \int_{T/3}^{2T/3} \mathrm{d}t \int_{0}^{T/3} G^{ba}(t-s) \mathrm{d}s \right\} \nonumber \\  
	   &\leq \eta^{ba} v_a v_b \left\{ -(1-\varepsilon) \frac{T^2}{18} + \frac{T^2}{9} \right\} \nonumber \\
       &= \eta^{ba} v_a v_b \frac{T^2}{18} (1 + \varepsilon) \: .
\label{eq:cost_asymm_cross}
\end{align}
Summing over all terms yields 
\begin{align}
C(\Pi) &\leq v^a v^b \frac{T^2}{18} (\eta^{ba} - \eta^{ab}) + 
			\varepsilon \frac{T^2}{18} \left( 2 \eta^{aa} v_a^2 + 2 \eta^{bb} v_b^2 + 3 \eta^{ab} v_a v_b + \eta^{ba} v_a v_b \right)\nonumber \\
       &= v^a v^b \frac{T^2}{36} (\eta^{ba} - \eta^{ab}) \nonumber \\
       &< 0
\label{eq:cost_asymm_bounded}
\end{align}
and we conclude that there is a price manipulation unless cross-impact is symmetric, i.e. we require $\eta^{ij} = \eta^{ji}$.

\end{proof}


\section{Power-law decay and impact}
\label{app:uniquedelta}

A popular class of unbounded decay kernels are power law kernels, e.g. $\bm{G}(\tau) \sim \tau^{-\gamma}, \: 0 < \gamma < 1$. Their advantage lies in allowing for more realistic parametrizations of the market impact function, such as concave power-law impact $\bm{f}(v) \sim \sgn (v) |v|^{\delta}$ for $0 < \delta < 1$. While \cite{gatheral2010no} establishes necessary conditions for such a model to be consistent in the one-dimensional case, numerical optimizations reported in \cite{curato2016discrete} find that violations of the principle of no-dynamic-arbitrage occur even when these conditions are verified, proofing them to be necessary but not sufficient. It  remains an open problem whether and under what conditions power law decay kernels and market impact functions are consistent. 
In this section we do not address this question but consider necessary constraints that arise from the presence of cross-impact. Specifically we show that in this case the shape parameter of the market impact function $\bm{f}$ needs to be unique for all self- and cross-impact terms.

\begin{lemma}
\label{lemma:uniquedelta}
Assume a price process as in (\ref{eq:price_ND}) where decay of market impact $\bm{G}(\tau)$ is a power law function, i.e. $G^{ij}(\tau) = \tau^{-\gamma^{ij}}$ with $0 < \gamma^{ij} < 1$ and $f^{ij}(v) = \eta^{ij} \sgn (v) |v|^{\delta^{ij}}$ is also power-law 
with $\eta^{ij} \geq 0$ for all $i,j$.\footnote{In the one-dimensional case \cite{gatheral2010no} shows that for absence of dynamic arbitrage it is also necessary that $\gamma \geq \gamma^* = 2- \frac{\log 3}{\log 2} \approx 0.415$ and $\gamma + \delta \geq 1$.}
Then absence of dynamic arbitrage requires that \begin{equation}
\delta^{ij} = \delta \quad \forall \: i,j \: . 
\label{eq:f_powerlaw}
\end{equation} 
\end{lemma}

\begin{proof}[Proof of Lemma \ref{lemma:uniquedelta}]

Consider a strategy of two phases lasting equally long where at first we build up a position at a constant trading rate from time $0$ until time $\Theta = T/2$ and then liquidate the position in a second phase from $T/2$ until $T$, i.e.
\begin{align}
\Pi = \left\{ \dot{\bm{x}}_t \right\} \quad, \quad \dot{\bm{x}}_t = \begin{cases} 
\left(v_{a}, v_{b} \right)^{\intercal} \quad &\text{for} \quad 0 \leq t \leq T/2 \\
\left(-v_{a}, -v_{b} \right)^{\intercal} \quad &\text{for} \quad T/2 < t \leq T 
\end{cases}.
\label{eq:strategy_specific_xx}
\end{align}
which is a special case of Example \ref{eg:strategy_easy}. Further we use the notation $\lambda = v_{a} / v_{b}$.  
Explicit calculation of the cost terms due to self-impact yields 
\begin{align}
C^{ii} &= \eta^{ii} v_i^{1+\delta^{ii}} \left[ \int_0^{T/2}{ \mathrm{d}t \int_0^t{\mathrm{d}s \tau^{-\gamma^{ii}} }} + \int_{T/2}^{T}{ \mathrm{d}t \int_{T/2}^t{\mathrm{d}s \tau^{-\gamma^{ii}} }} - \int_{T/2}^{T}{ \mathrm{d}t \int_0^{T/2}{\mathrm{d}s \tau^{-\gamma^{ii}} }} \right] \nonumber \\
	   &= \frac{\eta^{ii} v_i^{1+\delta^{ii}}}{1-\gamma^{ii}} \left[ 2 \int_0^{T/2}{  t^{1-\gamma^{ii}} \mathrm{d}t } - \int_{T/2}^{T}{ \left\{ (t-\frac{T}{2})^{1-\gamma^{ii}} - t^{1-\gamma^{ii}} \right\} \mathrm{d}t }  \right] \nonumber \\
       &= \frac{\eta^{ii} v_i^{1+\delta^{ii}} T^{2-\gamma^{ii}}}{(1-\gamma^{ii})(2-\gamma^{ii})}   \left[ 2^{\gamma^{ii}} -1 \right] \nonumber \\
       &=: \Lambda^{ii}(\eta^{ij}, \gamma^{ij}, T) v_i^{1+\delta^{ii}}
 \label{eq:cost_fpowerlaw_self}
\end{align}
and similarly we obtain for cross-impact
\begin{align}
C^{ij} &= - \eta^{ij} v_i v_j^{\delta^{ij}} \left[ \int_0^{T/2}{ \mathrm{d}t \int_0^t{\mathrm{d}s \tau^{-\gamma^{ij}} }} + \int_{T/2}^{T}{ \mathrm{d}t \int_{T/2}^t{\mathrm{d}s \tau^{-\gamma^{ij}} }} - \int_{T/2}^{T}{ \mathrm{d}t \int_0^{T/2}{\mathrm{d}s \tau^{-\gamma^{ij}} }} \right] \nonumber \\
       &= - \frac{\eta^{ij} v_i v_j^{\delta^{ij}} T^{2-\gamma^{ij}}}{(1-\gamma^{ij})(2-\gamma^{ii})}   \left[ 2^{\gamma^{ij}} -1 \right] \nonumber \\
       &=: - \Lambda^{ij}(\eta^{ij}, \gamma^{ij}, T) v_i v_j^{\delta^{ij}}
 \label{eq:cost_fpowerlaw_cross}
\end{align}
with $\Lambda^{ij} > 0 \: \forall \: ij$. Since $\lambda = v_a / v_b$ we substitute w.o.l.g. $v_a = \lambda v$, $v_b = v$. The total cost of the strategy is thus
\begin{align}
C &= \sum_{ij}{C^{ij}} \nonumber \\
 &= \Lambda^{bb} v^{1+\delta^{bb}} - \lambda^{\delta^{ba}} \Lambda^{ba} v^{1+\delta{ba}} - \lambda \Lambda^{ab} v^{1+\delta{ab}} + \lambda^{1+\delta^{aa}} \Lambda^{aa} v^{1+\delta^{aa}} \nonumber \\
 &= \Lambda^{bb} v^{1+\delta^{bb}} - \lambda^{\delta^{ba}} \Lambda^{ba} v^{1+\delta{ba}} + \mathcal{O}(\lambda)
 \label{eq:cost_fpowerlaw_sum}
\end{align}
where in the last step we are choosing $\lambda$ small enough so that linear terms in $\lambda$ can be neglected. Then if $\delta^{bb} > \delta^{ba}$ we can choose $v > 0$ small enough so that  $C < 0$ and likewise if $\delta^{bb} < \delta^{ba}$ we can choose $v$ large enough so that there is a price-manipulation. Therefore we require $\delta^{bb} = \delta^{ba} =: \delta^b$ and likewise $\delta^{aa} = \delta^{ab} =: \delta^a$. Therefore we can re-express the cost as 
\begin{align}
C &= \Lambda^{bb} v^{1+\delta^{b}} - \lambda^{\delta^{b}} \Lambda^{ba} v^{1+\delta{a}} - \lambda \Lambda^{ab} v^{1+\delta{a}} + \mathcal{O} ( \lambda^{1+\delta^{a}} ) \nonumber \\
 &= \underbrace{(\Lambda^{bb} - \lambda^{\delta^b} \Lambda^{ba})}_{:= \Lambda^b} v^{1+\delta^{b}} - \lambda^{\delta^{a}} \Lambda^{ba} v^{1+\delta{a}} + \mathcal{O}(\lambda^{1+\delta^a})
 \label{eq:cost_fpowerlaw_sum2}
\end{align}
now also considering linear terms in $\lambda$. Since $\lambda$ is small we can use that $\Lambda^b>0$ and by the same arguments as above we conclude that absence of arbitrage requires $\delta^a = \delta^b = \delta$. 

\end{proof}

\section{List of ISINs}
\label{app:isins}
\renewcommand{\floatpagefraction}{0.8}

Table \ref{tab:bonds} reports basic descriptive statistics and liquidity measures for all the $N=33$ bonds that were used for estimations. The set of bonds was selected as all fixed rate or zero-coupon Italian sovereign bonds with at least 5,000 trades throughout our sample. Note that some bonds were issued during our sample period and therefore less than 194 trading days were observed.

The description for each bond lists the bond type (BTPs refers to fixed-income treasury bonds, CTZ to zero coupon bonds), the fixed interest rate (where applicable) and the maturity date. Maturity is the time from issuance of a bond to the maturity date in years and time-to-maturity is calculated as the remaining time from the end of our sample (October 16, 2015) to the maturity date. The trade volume measures in Table \ref{tab:bonds} (mean traded volume per day and mean volume per trade) are reported as face volume traded and to arrive at the value one has to multiply by the price. The liquidity measures of the limit order book (mean number of limit order book updates per day, mean spread and ratio of tick size over mean spread) are computed from 10:00 to 17:00 of each day as we restrict our analysis to this period to avoid intraday seasonalities. Bid-ask spread is given in units of basis points of the face value, e.g. an average spread $14.9 \text{bp}$ corresponds to a contract with a standard face value of EUR $100$ offered at a mean spread of $14.9$ euro-cents.

\floatpagestyle{plain}
\begin{table}[t]
\begin{adjustbox}{max width=1.0\textwidth}
\centering
\begin{tabular}{l|l|rrrrrrrrr} 
ISIN & description & \rot{maturity (in years)} & \rot{time-to-maturity (in years) } & \rot{\# days with observations } & \rot{ avg. \# trades per day } & \rot{ avg traded volume per day (in million EUR) } & \rot{ avg. volume per trade (in 1,000 EUR) } & \rot{ avg. \# LOB updates per day (in 1,000)} & \rot{ avg. spread (in basis points) } & \rot{ tick size / avg. spread}\\  
  \hline
IT0001278511 & BTPS 5.250  01/11/29 & 31.0 & 14.1 & 194 & 108.8 & 5.5 & 50.3 & 31.4 & 14.9 & 0.07 \\ 
  IT0003535157 & BTPS 5.000  01/08/34 & 31.0 & 18.8 & 194 & 103.3 & 5.5 & 52.9 & 44.3 & 22.7 & 0.04 \\ 
  IT0003934657 & BTPS 4.000  01/02/37 & 31.5 & 21.3 & 194 & 882.1 & 66.5 & 75.4 & 31.0 & 5.5 & 0.18 \\ 
  IT0004009673 & BTPS 3.750  01/08/21 & 15.5 & 5.8 & 194 & 84.4 & 5.8 & 68.5 & 18.5 & 5.6 & 0.18 \\ 
  IT0004019581 & BTPS 3.750  01/08/16 & 10.4 & 0.8 & 194 & 29.2 & 1.8 & 60.5 & 4.9 & 2.3 & 0.04 \\ 
  IT0004164775 & BTPS 4.000 01/02/17 & 10.1 & 1.3 & 194 & 22.8 & 1.5 & 67.8 & 3.2 & 3.7 & 0.03 \\ 
  IT0004361041 & BTPS 4.500  01/08/18 & 10.3 & 2.8 & 194 & 23.3 & 2.6 & 111.9 & 6.8 & 5.1 & 0.20 \\ 
  IT0004423957 & BTPS 4.500  01/03/19 & 10.5 & 3.4 & 194 & 22.4 & 2.4 & 108.1 & 10.8 & 5.6 & 0.18 \\ 
  IT0004489610 & BTPS 4.250  01/09/19 & 10.3 & 3.9 & 194 & 42.7 & 4.1 & 95.6 & 12.2 & 5.7 & 0.18 \\ 
  IT0004532559 & BTPS 5.000  01/09/40 & 31.0 & 24.9 & 194 & 179.3 & 10.6 & 59.1 & 35.0 & 17.6 & 0.06 \\ 
  IT0004536949 & BTPS 4.250  01/03/20 & 10.4 & 4.4 & 194 & 41.2 & 6.7 & 162.8 & 14.3 & 5.9 & 0.17 \\ 
  IT0004594930 & BTPS 4.000  01/09/20 & 10.4 & 4.9 & 194 & 57.6 & 5.1 & 89.3 & 15.9 & 5.8 & 0.17 \\ 
  IT0004634132 & BTPS 3.750  01/03/21 & 10.5 & 5.4 & 194 & 49.3 & 4.2 & 85.3 & 18.0 & 6.1 & 0.16 \\ 
  IT0004695075 & BTPS 4.750  01/09/21 & 10.5 & 5.9 & 194 & 25.3 & 2.3 & 89.7 & 17.8 & 6.7 & 0.15 \\ 
  IT0004759673 & BTPS 5.000  01/03/22 & 10.5 & 6.4 & 194 & 30.7 & 3.1 & 99.9 & 18.1 & 7.5 & 0.13 \\ 
  IT0004801541 & BTPS 5.500 01/09/22 & 10.5 & 6.9 & 194 & 25.2 & 2.3 & 91.4 & 19.0 & 8.3 & 0.12 \\ 
  IT0004848831 & BTPS 5.500 01/11/22 & 10.2 & 7.0 & 194 & 20.9 & 2.2 & 107.2 & 16.8 & 8.7 & 0.11 \\ 
  IT0004898034 & BTPS 4.500  01/05/23 & 10.2 & 7.5 & 194 & 35.1 & 3.6 & 101.4 & 18.8 & 8.3 & 0.12 \\ 
  IT0004923998 & BTPS 4.750  01/09/44 & 31.3 & 28.9 & 194 & 252.9 & 18.7 & 73.9 & 33.1 & 15.4 & 0.07 \\ 
  IT0004953417 & BTPS 4.500  01/03/24 & 10.6 & 8.4 & 194 & 50.9 & 9.2 & 180.0 & 23.4 & 5.5 & 0.18 \\ 
  IT0005001547 & BTPS 3.750  01/09/24 & 10.5 & 8.9 & 194 & 47.1 & 7.3 & 154.7 & 22.5 & 6.5 & 0.15 \\ 
  IT0005024234 & BTPS 3.500  01/03/30 & 15.8 & 14.4 & 194 & 257.3 & 22.6 & 87.8 & 27.2 & 9.3 & 0.11 \\ 
  IT0005028003 & BTPS 2.150  15/12/21 & 7.5 & 6.2 & 194 & 83.1 & 14.4 & 173.0 & 16.7 & 4.7 & 0.21 \\ 
  IT0005030504 & BTPS 1.500  01/08/19 & 5.1 & 3.8 & 194 & 25.8 & 6.7 & 258.5 & 10.5 & 5.4 & 0.18 \\ 
  IT0005044976 & CTZ 14- 30/08/16 24M & 2.0 & 0.9 & 194 & 31.8 & 4.3 & 135.3 & 2.5 & 1.7 & 0.06 \\ 
  IT0005045270 & BTPS 2.500 01/12/24 & 10.3 & 9.1 & 194 & 202.1 & 29.0 & 143.7 & 25.7 & 4.4 & 0.23 \\ 
  IT0005069395 & BTPS 1.050  01/12/19 & 5.0 & 4.1 & 194 & 74.6 & 14.2 & 190.3 & 11.2 & 4.2 & 0.24 \\ 
  IT0005083057 & BTPS 3.250  01/09/46 & 31.6 & 30.9 & 163 & 892.4 & 78.3 & 87.7 & 31.4 & 6.1 & 0.16 \\ 
  IT0005086886 & BTPS 1.350 15/04/22 & 7.2 & 6.5 & 145 & 124.1 & 16.5 & 133.1 & 18.0 & 4.7 & 0.21 \\ 
  IT0005090318 & BTPS 1.500  01/06/25 & 10.3 & 9.6 & 135 & 332.7 & 36.8 & 110.7 & 28.3 & 4.3 & 0.23 \\ 
  IT0005094088 & BTPS 1.650  01/03/32 & 17.0 & 16.4 & 134 & 603.3 & 41.4 & 68.7 & 22.9 & 6.1 & 0.16 \\ 
  IT0005107708 & BTPS 0.700  01/05/20 & 5.0 & 4.5 & 121 & 65.5 & 11.7 & 178.6 & 12.3 & 5.6 & 0.18 \\ 
  IT0005127086 & BTPS 2.000  01/12/25 & 10.3 & 10.1 & 35 & 174.8 & 14.6 & 83.3 & 18.6 & 3.8 & 0.26 \\ 
   \hline
\end{tabular}
\end{adjustbox}
\caption{Descriptives and liquidity measures for the set of bonds used in estimation. }
\label{tab:bonds}
\end{table}

%
\end{document}